\pdfoutput=1

\documentclass[lettersize,journal]{IEEEtran}
\usepackage{amsmath,amsfonts}
\usepackage{algorithmic}
\usepackage{algorithm}
\usepackage{array}
\usepackage[caption=false,font=normalsize,labelfont=sf,textfont=sf]{subfig}
\usepackage{textcomp}
\usepackage{stfloats}
\usepackage{url}
\usepackage{verbatim}
\usepackage{graphicx}
\usepackage{cite}
\usepackage{booktabs}
\usepackage{supertabular}
\usepackage{amsfonts,amssymb}
\usepackage{mathrsfs}
\usepackage{diagbox}
\usepackage{multirow} 
\usepackage{mathtools}
\usepackage{amssymb}
\usepackage{enumitem}
\usepackage{threeparttable}
\usepackage{url}
\usepackage{hyperref}
\usepackage{ulem}

\newcommand{\believes}{\mid\!\equiv}
\newcommand{\sees}{\triangleleft}
\newcommand{\oncesaid}{\mid\!\sim}
\newcommand{\controls}{\mid\!\Rightarrow}
\newcommand{\fresh}[1]{\#(#1)}

\newcommand{\encrypt}[2]{{ \{ #1 \} }_{#2}}
\newcommand{\sharekey}[1]{\xleftrightarrow{#1}}

\newcommand{\secret}[1]{\xleftrightharpoons{#1}}

\newtheorem{theorem}{Theorem}[section]

\newenvironment{proof}[1][Proof]{\begin{trivlist}
\item[\hskip \labelsep {\bfseries #1}]}{\end{trivlist}}

\newcommand{\qed}{\nobreak \ifvmode \relax \else
      \ifdim\lastskip<1.5em \hskip-\lastskip
      \hskip1.5em plus0em minus0.5em \fi \nobreak
      \vrule height0.75em width0.5em depth0.25em\fi}

\newlist{steps}{enumerate}{1}
\setlist[steps, 1]{listparindent=1.5em ,align=left, leftmargin= 5 pt,label = Step \arabic*:}

\DeclareMathSizes{10}{10}{5}{4.5}

\begin{document}

\title{BEPHAP: A Blockchain-Based Efficient Privacy-Preserving Handover Authentication Protocol with Key Agreement for Internet of Vehicles}

\author{Xianwang Xie,~\IEEEmembership{student member,~IEEE,}, Bin Wu,~\IEEEmembership{member,~IEEE,}, Botao Hou,~\IEEEmembership{student member,~IEEE,}

\thanks{Manuscript received XXX xx, xxxx; revised XXX xxx, xxxx. The work was supported in part by the National Natural Science Foundation of China under Grant U1936119, and Grant 62272007 and in part by the Major Science and Technology Project of Hainan Province under Grant ZDKJ2019003, and in part by the Key Projects of Science and Technology Of China State Railway Group Co.,Ltd under Grant N2021W003. \textit{(Corresponding author: Bin Wu.)}}

\thanks{Xianwang Xie is with the State Key Laboratory of Information Security, Institute of Information Engineering, CAS, Beijing, China, and also with the School of Cyber Security, University of Chinese Academy of Sciences, Beijing, China.}

\thanks{Bin Wu is with the State Key Laboratory of Information Security, Institute of Information Engineering, CAS, Beijing, China, and also with the School of Cyber Security, University of Chinese Academy of Sciences, Beijing, China (e-mail: wubin@iie.ac.cn).}

\thanks{Botao Hou is with the State Key Laboratory of Information Security, Institute of Information Engineering, CAS, Beijing, China, and also with the School of Cyber Security, University of Chinese Academy of Sciences, Beijing, China.}

}

\markboth{\LaTeX\ Class Files,~Vol.~XX, No.~XX, XXX~202X}%
{Shell \MakeLowercase{\textit{et al.}}: A Sample Article Using IEEEtran.cls for IEEE Journals}


\maketitle

\begin{abstract}
The Internet of Vehicles (IoV) can significantly improve transportation efficiency and ensure traffic safety. Authentication is regarded as the fundamental defense line against attacks in IoV. However, the state-of-the-art approaches suffer from several drawbacks, including bottlenecks of the single cloud server model, high computational overhead of operations, excessive trust in cloud servers and roadside units (RSUs), and leakage of vehicle trajectory privacy.
In this paper, BEPHAP, a Blockchain-based Efficient Privacy-preserving Handover Authentication Protocol with key agreement for internet of vehicles, is introduced to address these problems. BEPHAP achieves anonymous cross-domain mutual handover authentication with key agreement based on the tamper-proof blockchain, symmetric cryptography, and the chameleon hash function under a security model that cloud servers and RSUs may launch attacks. 
BEPHAP is particularly well suited for IoV since it allows vehicles only need to perform lightweight cryptographic operations during the authentication phase.
BEPHAP also achieves data confidentiality, unlinkability, traceability, non-repudiation, non-frameability, and key escrow freeness. Formal verification based on ProVerif and formal security proofs based on the BAN logic indicates that BEPHAP is resistant to various typical attacks, such as man-in-the-middle attacks, impersonation attacks, and replay attacks. Performance analysis demonstrates that BEPHAP surpasses existing works in both computation and communication efficiencies. And the message loss rate remains 0 at 5000 requests per second, which meets the requirement of IoV.

\end{abstract}


\section{Introduction}
Internet of Vehicles (IoV), as an essential part of the Intelligent Transportation Systems (ITS), can significantly improve the efficiency of transportation and reduce traffic accidents and energy consumption \cite{DBLP:journals/fgcs/LiangLWCLZ19, DBLP:journals/jsa/BaggaSDV21,DBLP:journals/tits/YangWYHS22}.
IoV consists of participating vehicles with On-board Units (OBUs), roadside units (RSUs), and cloud servers. The roadside units are infrastructures deployed along the road, which can be used as edge servers to interact with the vehicles, the cloud servers are responsible for providing services to vehicles, and the communication between the vehicles and the cloud servers is realized through the roadside units. In such a network, vehicles can exchange real-time traffic information with other entities \cite{DBLP:journals/tvt/HuangYC11}, such as location, speed, traffic congestion, etc. To effectively and securely support information dissemination, reliable authentication schemes are indispensable. In an authentication scheme, two parties, which can be seen as a vehicle and RSU, can confirm whether the other party is legitimate with several message exchanges in an insecure communication channel \cite{DBLP:journals/jpdc/XuLLXJ21}.
However, several challenges must be handled in IoV:

1) User privacy may be obtained by adversaries who perform privacy mining and data association in the massive real-time messages dissemination of IoV. A natural idea to address privacy leakage is adopting a pseudonyms mechanism to protect the real identity\cite{DBLP:journals/tvt/LiuWC15}. However, multiple pseudonyms for the same vehicle may be linked and associated by powerful adversaries by monitoring spatiotemporal relationships\cite{DBLP:conf/globecom/HaoCR08}, and the vehicle's trajectory may leak since the static road topology restricts the movement of the vehicle. While, vehicle owners usually do not want their private information (such as their real identities and driving trajectory) to be revealed. Therefore, the privacy of the vehicle (owner) should be protected by authentication protocols. However, vehicles with malicious behavior should be identified and punished, so privacy protection should be conditional \cite{7840036}.

2) The openness of IoV makes it vulnerable to various security threats, such as typical replay attacks, impersonation attacks, eavesdropping attacks, tampering attacks, and man-in-the-middle attacks. Many IoV authentication schemes have been proposed to deal with these security threats. However, most have not been verified by formal security verification tools. It's worth noting that most state-of-the-art approaches assume cloud servers and RSUs are trusted entities. While in reality, cloud servers and RSUs are taken care of by different parties. They may be curious about user privacy and act as a passive attacker to cause the leakage of confidential information. And they may also launch an active attack to frame honest vehicles \cite{DBLP:journals/tifs/FengSXW21} while tracing malicious ones.

3) The communication time among RSUs and vehicles is restricted due to the vehicles' high speed (e.g., 36–140 km/h) and short communication range (e.g., 100–300 m)\cite{DBLP:journals/tits/JiangZW16}. Moreover, an RSU should verify approximately 5000 messages per second with hundreds of vehicles in its coverage because the transmission frequency of the traffic-related messages can exceed 10 times per second \cite{DBLP:journals/tifs/FengSXW21}. Thus, a low-latency authentication protocol is a must in IOV.

4) In most existing IoV solutions, a centralized structure is adopted, which means all vehicles can only authenticate with the cloud server, and the RSU just acts as an intermediary node to facilitate communication between the vehicle and the cloud server. For such a centralized architecture, as the number of vehicles increases, the computing and communication resource bottleneck of the central server may make it fail to accomplish mutual authentication with all vehicles in the network within a limited time. Therefore, a multi-cloud network model should be used in IoV \cite{DBLP:journals/jpdc/XuLLXJ21}. However, due to vehicles' long-distance mobility, vehicles need to be capable of performing cross-domain authentication under the multi-cloud network model \cite{DBLP:journals/jnca/LiXMW12}.




In this paper, BEPHAP, a Blockchain-based Efficient Privacy-preserving Handover Authentication Protocol with key agreement for IoV, is proposed to solve the above problems. The blockchain, as a distributed peer-to-peer network, is suitable for addressing cross-domain authentication problems in multi-cloud model. In BEPHAP, the blockchain is used to synchronize the vehicle-related information in each cloud server and enable them to manage the vehicles' information in the network jointly. Due to the tamper-proof property of the blockchain, any attacker, including cloud servers and RSUs, cannot easily tamper with the vehicle-related information stored in the blockchain.

The main contributions of this paper are as follows:
\begin{enumerate}
    \item To the best of our knowledge, BEPHAP is the first blockchain-based authentication protocol scheme for IoV that simultaneously implement mutual authentication with key agreement, data confidentiality, identity anonymity, unlinkability, traceability, non-repudiation, non-frameability, key escrow freeness \cite{DBLP:journals/tifs/YangZZCZ22}, cross-domain, formal security proof, and verification by formal security verification tools.
    \item A novel low-latency authentication scheme is proposed, suitable for IoV scenarios with limited vehicle computing resources. Because BEPHAP allows vehicles only need to perform lightweight cryptographic operations in the authentication phase, such as hash and symmetric encryption.
    \item We consider a security model in which cloud servers and RSUs may launch attacks. Under the security model, the security evaluation shows that BEPHAP can provide conditional privacy-preserving and prevent honest vehicles from being framed by any entities, including cloud servers or RSU.
    \item It is proved that our protocol possesses various security properties, based on BAN logic \cite{DBLP:conf/sosp/BurrowsAN89} and ProVerif \cite{blanchet2018proverif} formal security verification tools.
    
\end{enumerate}

The remainder of this paper is organized as follows. Section \ref{Related Work} introduces related works of the proposed research. Section \ref{Preliminaries and System Overview} introduces the preliminaries, system model, and secure model. We elaborate on the proposed scheme in Section \ref{Proposed Scheme} and present the security evaluation in Section \ref{Security Evaluation}. In Section \ref{Functionality and Performance Evaluation}, we evaluate the functionality and performance of our model and compare it with the existing schemes. Finally, Section \ref{Conclusion} concludes the paper.

\section{Related Work\label{Related Work}}

In 2005, Choi et al. \cite{DBLP:conf/mswim/ChoiJW05} combined symmetric authentication with short-term pseudonyms in IoV. Since symmetric cryptography has higher computational efficiency and lower communication overhead, authentication efficiency is improved. In this scheme, vehicles generate short-term pseudonyms from unique identifiers and seed values received from authorities. The vehicle and the RSU share a secret key, so the RSU can verify that the vehicle's identity is legitimate by verifying that the vehicle has the secret key. However, this scheme has the problems of vulnerable key management and lack of non-repudiation.

To protect the identity privacy of the vehicle, Raya et al. \cite{raya2007securing} proposed a PKI-based privacy protection authentication scheme for IoV. In this scheme, certificates are issued and managed by a certificate authority. The scheme protects the identity privacy of the vehicle through a pseudonymous certification. However, the vehicle needs to continuously maintain a certificate revocation list and pre-install a large number of public-private key pairs and corresponding certificates, which will cause a lot of computational overhead and storage burden to the vehicle. Lu et al. \cite{DBLP:conf/infocom/LuLZHS08} proposed a protocol that allows vehicles to request a temporary certificate from the RSU, thereby solving the problem of certificate pre-storage. However, it involves many bilinear pairing operations causing high computational costs. Lu et al. \cite{DBLP:journals/tvt/LuLLLS12} proposed a pseudonym update strategy to prevent being tracked by limiting the pseudonym's lifetime. When vehicles gather in parking lots or road intersections, if the anonymity set size reaches a threshold, the vehicles can change pseudonyms at the same time. But this pseudonym-updating strategy does not perform well in scenarios with low vehicle density. Wang et al. \cite{DBLP:journals/comcom/WangY17} proposed an anonymous authentication scheme. The certificates of the vehicle and the RSU are distributed by a trusted organization, and the RSU authenticates the vehicle based on the long-term certificate of the vehicle and assigns a master key to the certified vehicle. The vehicle can generate a pseudonym by itself through the master key, thus reducing the load of the trusted authority. However, this scheme does not satisfy the unlinkability and cannot protect the trajectory privacy of the vehicle.

An authentication scheme based on IBS (Identity-Based Signature), proposed by Shamir et al. \cite{DBLP:conf/crypto/Shamir84}, can be used to solve the general shortcomings of PKI-based authentication schemes, that is, the computational, communication, and storage overheads caused by certificates and revocation lists. 
The private key of the vehicle is generated by the Key Generation Center (KGC) according to the vehicle's identity, and the vehicle's identity information is used as the public key. Zhang et al. \cite{DBLP:conf/infocom/ZhangLLHS08,DBLP:journals/winet/ZhangHT11} implemented lightweight message authentication and privacy protection based on IBS in the scenario of IoV. However, Lee et al. \cite{DBLP:journals/winet/LeeL13} pointed out that the scheme of Zhang et al. \cite{DBLP:conf/infocom/ZhangLLHS08,DBLP:journals/winet/ZhangHT11} does not achieve non-repudiation and is vulnerable to replay attacks. And by adding message signature, message verification, and anonymous identity generation, an improved scheme is proposed, which can achieve efficient batch authentication and solve the problems of Zhang et al. \cite{DBLP:conf/infocom/ZhangLLHS08,DBLP:journals/winet/ZhangHT11}. But an increase in the number of invalid signatures may degrade the performance of this scheme. What's more, Bayat et al. \cite{DBLP:journals/winet/BayatBRA15} pointed out that the scheme of Lee et al. \cite{DBLP:journals/winet/LeeL13} is vulnerable to impersonation attacks. An attacker can imitate a legitimate vehicle to generate a valid signature, thereby sending false messages. Shim et al. \cite{DBLP:journals/tvt/Shim12} proposed a conditional privacy-preserving authentication scheme based on bilinear pairings, but pairing operations make the scheme computationally expensive. And based on the IBS scheme, the vehicle's private key is generated by the key generation center. In other words, the key generation center knows the private key of each vehicle. That is, there is a key escrow problem.

To solve the certificate management problem in the PKI-based scheme and the key escrow problem in the IBS-based scheme, Tsai et al. \cite{DBLP:journals/sj/Tsai17} proposed a certificateless scheme. But due to the use of bilinear pairing operations, the computational overhead is high. Ming and Cheng \cite{ming2019efficient} proposed a conditional privacy-preserving authentication scheme with low transmission overhead. The scheme can be proved to be secure under the random oracle model. However, its transmission overhead is still too high to meet the requirement of IoV. Yang et al. \cite{DBLP:journals/tifs/YangZZCZ22} designed a certificateless aggregation signcryption scheme. However, since it is also implemented using bilinear pairing, the computational overhead is still too high. Additionally, it does not achieve mutual authentication.

Blockchain, as an emerging technology in recent years, has the characteristics of distribution, non-tampering, and traceability. 
Arora et al. \cite{arora2018block} proposed a blockchain-based authentication protocol for the IoV. This scheme uses RSU as part of the blockchain, which is obviously inappropriate because RSU has limited storage resources and will face more security risks. 
Wang et al. \cite{DBLP:journals/jsa/WangHXLMZ21} proposed a blockchain-assisted handover authentication and key agreement scheme in a multi-server edge computing environment, but this scheme is not oriented to the scenario of the IoV, and it is not unframeable. Feng et al. \cite{DBLP:journals/jsa/FengSXL21} proposed an efficient privacy-preserving authentication model (EPAM) for the IoV, which uses asynchronous accumulators to extend the blockchain. However, it is not unframeable because the vehicle's certificate can be placed in other messages by RSM. Xu et al. \cite{DBLP:journals/jpdc/XuLLXJ21} proposed an efficient authentication protocol for the IoV based on blockchain and symmetric encryption. The authentication process only has low-overhead operations such as hashing and XOR, and the protocol transfers the computational load of the authentication server to the RSU, thereby improving the authentication efficiency. But this scheme cannot protect the privacy of the vehicle. Zhang et al. \cite{DBLP:journals/tdsc/0002DB021} proposed a robust, general handover authentication protocol for 5G environments. Using the consensus and anti-tampering capabilities of the blockchain, performing handover authentication between heterogeneous access networks in different domains is realized. However, in this scheme, although the pseudo-identifiers in each message are different, the attacker can still calculate the chameleon hash according to the messages. By comparing the chameleon hashes, different messages can be linked to the same device and then infer the location privacy of the device based on trajectory mobility.

\section{Preliminaries and System Overview\label{Preliminaries and System Overview}}
\subsection{Elliptic Curve Cryptosystem\label{Elliptic Curve Cryptosystem}}
Let $p$ be a prime number, the finite field $\mathbb{F}_p$ is determined by the prime number $p$. Let $E(\mathbb{F}_p)$ be an elliptic curve over the finite field $\mathbb{F}_p$. Let $a, b \in \mathbb{F}_p$. Define an  elliptic curve $E(\mathbb{F}_p):y^2=x^3+ax+b$ mod $p$. Let $O$ be the point at infinity, $P$ a point of $E(\mathbb{F}_p)$ with prime order $q$, $\mathbb{G}$ an additive elliptic curve group consisting of $O$ and other points on $E(\mathbb{F}_p)$ with generator $P$. The elliptic curve group $\mathbb{G}$ has the following hardness assumptions and properties\cite{DBLP:journals/joc/MenezesV93}.

\begin{itemize}

\item Additive operation: Let $P$ and $Q$ be two points of the additive elliptic curve group $\mathbb{G}$. We can get $R=P+Q$. If $P=Q$, then $R=2P$. If $P=-Q$, then $R=O$, If $P\neq Q$ and $P\neq -Q$, then R is the intersection of $E(\mathbb{F}_p)$ and the straight line connecting P and Q.

\item Scalar point multiplication: Let $m \in \mathbb{Z}_q^*, P \in \mathbb{G}$, the scalar multiplication of $E(\mathbb{F}_p)$ is defined as $m \cdot P=P+P+...+P$.

\item Elliptic curve discrete logarithm problem (ECDLP): For a probabilistic polynomial-time (PPT) adversary, it is computationally hard to calculate $x \in \mathbb{Z}_q^*$ in the case of known two points $P$, $Q=xP \in \mathbb{G}$.


\end{itemize}

\subsection{Chameleon Hash Function\label{Chameleon Hash Function}}
Chameleon hash function\cite{DBLP:conf/ndss/KrawczykR00}, also called trapdoor-hash function, is hash function featuring a trapdoor. The knowledge of the trapdoor allows one to find arbitrary collisions in the domain of the function. Let $(m^*, r^*)$ be an initial input where $m^*,r^* \in \mathbb{Z}_q^*$, $(k,x)$ the trapdoor satisfying $x\in \mathbb{Z}_q^*$ and $k=m^*+r^*x$, $(P,Y)$ the hash key where $P$ is a point of $E(\mathbb{F}_p)$ with prime order $q$, $Y=xP$. Then the ECC variant of the chameleon hash function is defined as $CH_Y(m,r)=mP+rY$ where $m,r \in \mathbb{Z}_q^*$, $m=k-rx$ mode $q$. The following properties are owned by the chameleon hash function:

\begin{itemize}

\item Collision Resistance: It's infeasible for any probabilistic polynomial-time (PPT) adversary except the holder of the trapdoor to compute $m',r' \in \mathbb{Z}_q^*$ such that $(m,r) \neq (m',r')$ satisfying $CH_Y(m,r)=CH_Y(m',r')$.

\item Trapdoor Collisions: Given an input $r' \in \mathbb{Z}_q^*$, the holder of the trapdoor can easily calculate $m'=k-r'x$ mod $q$ such that $CH_Y(m',r')=CH_Y(m^*,r^*)$, where $(m^*,r^*)$ is the initial input.

\end{itemize}

\subsection{Pseudo-Random Function\label{pseudo-random function}}
A pseudo-random function (PRF) is an efficient and deterministic algorithm taking two inputs and returning a pseudorandom output sequence. For example, take a key $k$ and a binary string $x \in \mathcal{X}$ as input, where $\mathcal{X}$ denotes the input space, we can get a binary string $y=PRF(k,x) \in \mathcal{Y}$, where $\mathcal{Y}$ stands for the output space. A PRF is secure if any PPT adversary cannot distinguish the output of the secure PRF from that of random function\cite{DBLP:conf/crypto/GoldreichMW86}.

\subsection{System Model}

Our system architecture is illustrated in \ref{Scenes schematic}, in which there exist six components as explained below:

\begin{itemize}
    \item Law Enforcement Authority (LEA): LEA should be an institution authorized by law that can trace a malicious vehicle's real identity for auditing the network. LEA is the only authority with the ability to reveal the real identity of a hostile vehicle. LEA can be seen as a cloud server with considerable computing and storage resources that can deploy a full blockchain node, enabling it to send registration information to the blockchain. The LEA is responsible for managing the vehicles' real identity and registration information.
    
    \item Regional Service Manager (RSM): RSMs, as cloud servers with strong computing and storage capabilities, are delegated by the LEA to provide registration, authentication, and revocation of vehicles. The entire network consists of multiple domains, and each RSM is responsible for serving vehicles in one domain in the whole network. Each RSM acts as a full node in the blockchain network storing registration, authentication, and revocation information of vehicles. In addition, to punish or track malicious vehicles, RSM is obliged to report the trajectory of malicious vehicles in its domain to LEA. Thus, RSM needs the ability to link multiple messages of malicious vehicles.
    
    \item Roadside Unit (RSU): RSUs are widely distributed alongside the roads to optimally organize and coordinate vehicular communications. The calculation and storage capability of an RSU are weaker than an RSM. While considering the development of hardware, the RSU still has considerable computing and storing resources to undertake part of the load during the authentication process, thereby reducing the overhead of the RSM. Besides, to punish or track hostile vehicles, RSU is obliged to report the trajectory of malicious vehicles in its service area to RSM. Hence, RSU needs the capability to link multiple messages of the hostile vehicle.
    
    \item On-board Unit (OBU): OBU is the communication and computing unit deployed on a vehicle with limited computing and storage capacity. The vehicle with OBU can communicate with infrastructure or other vehicles with the help of surrounding RSUs.
    
    \item Fog Sever (FS): FS, whose computing and storing resources are less than RSM, is responsible for forwarding messages between RSM and RSU. Vehicles usually move in a wide range, possibly involving many different regions. In reality, different regions usually differ in the deployment of authentication servers, the type of access network, and the level of access points. For example, in some regions, there are several fog servers between RSM and RSU, while in others there is no fog server.
    
    \item Blockchain: The blockchain is responsible for synchronizing the registration list and revocation list. Each RSM can provide registration and authentication services for vehicles or revoke a malicious vehicle in a distributed manner with the help of blockchain. Registration and revocation information is stored on the blockchain and shared by LEA and all RSMs. Different domains can share registration and revocation information through the blockchain, although there are differences in network type and deployment level. Once a vehicle is registered with an RSM, it can verify identity in all RSMs' responsible areas. Similarly, when a vehicle with malicious behavior is discovered and revoked by an RSM, the vehicle cannot be authenticated successfully in all RSMs' responsible areas. With the help of blockchain, we can achieve distributed cross-domain authentication. In addition, we can use blockchain to prevent honest vehicles from being framed by compromised LEA in the process of pursuing accountability.
    
\end{itemize}

\subsection{Secure Model}
\subsubsection{Secure Assumption}

Regarding LEA, RSM, RSU, and FS, for their reputation, they will not take the initiative to disrupt the protocol process and cause authentication and key negotiation to fail. Vehicles also won't do it for their benefit. They also won't leak the initial private information, such as the real identity of vehicles owned by the LEA for revealing malicious vehicles' identities and the driving traces of vehicles that the LEA, RSU, and RSMs can obtain for tracking hostile vehicles. Only the LEA in the model knows the real identity of the vehicles.


Any entity in the model must strictly protect its own private key.
LEA, RSM, RSU, and FS communicate through wired connection due to fixed geographical location, so it is assumed that the communication between LEA, RSM, RSU is secure.

\subsubsection{Threat Model\label{Threat Model}}

Generally, there are external attackers and internal attackers,
the former refers to the entities not directly involved in IoV, and the latter relates to entities directly engaged in IoV, i.e., vehicles, RSUs, RSMs, and LEA.

External adversaries can launch passive and active attacks. Regarding passive attacks, an external adversary may act as a passive listener keeping monitoring public communication channels, trying to get some confidential information, e.g., the real identity of vehicles, the driving traces of vehicles, and the plaintexts of encrypted messages. Regarding active attacks, as defined in the Dolev-Yao threat model \cite{DBLP:journals/tit/DolevY83}, an external adversary can read, intercept, fabricate, modify, and replay the transmitted data packets over the channel.

There is a significant difference between internal adversaries and external adversaries. That is, an internal attacker will honestly execute the authentication and key agreement protocol without actively leaking the initial privacy. However, they may still launch passive and active attacks. As for passive attacks, an internal adversary may perform as a passive listener monitoring public or secure communication channels, trying to get confidential information other than the initial confidential information it should possess. For instance, some RSM, RSU, and FS may be curious about vehicles' real identities. Concerning active attacks, an internal attacker may be bribed or want to collude with malicious vehicles for profit, so they may frame honest vehicles that conflict with their interests or cover malicious vehicles. Precisely, an internal adversary may fabricate, modify, and replay the authentication information of an honest vehicle to send malicious messages or output the honest vehicle's identity in the process of tracing the hostile vehicle's real identity, thereby framing the honest or covering the hostile vehicle.


\subsubsection{Secure Requirements}
BEPHAP should satisfy the following security requirements.
\begin{itemize}

\item Mutual authentication correctness and integrity: For the correctness property, both authorized parties, such as the authorized vehicle and RSU, communicating with BEPHAP can be verified that they are indeed legitimate entities. The message sent by an accredited vehicle or RSU can be proved correct without being modified or fabricated.

\item Data Confidentiality: Confidentiality guarantees that only the legitimate receivers can learn the content of the message sent by a vehicle that may contain sensitive information arousing the curiosity of attackers. External adversaries can't get any confidential information from the message and no internal attacker can obtain confidential information other than the initial private information.

\item Conditional privacy-preserving: BEPHAP meets the conditional privacy-preserving requirements as follows:
\begin{itemize}
\item Identity anonymity: no one except for LEA can reveal the real identity of a vehicle according to the received or intercepted message sent by the vehicle.

\item Unlinkability: No one can link multiple messages to the same vehicle, except for entities that need this information initially, such as LEA, RSM, and RSU.
\end{itemize}

\item Traceability: The real identity of the message sender should be bound to the message so that LEA can trace the vehicle sending malicious messages.

\item Non-repudiation: Only the LEA can identify the vehicle related to the verification message. It can reveal the real identity of a vehicle in the event of a dispute. Only the vehicle itself can send the associated authentication information. Therefore, when the real identity of a hostile vehicle is traced through the authentication information, the hostile vehicle cannot deny its malicious behavior.

\item Non-frameability: No external or internal attacker can frame an honest vehicle, even the LEA.

\item Resisting attacks: Our proposed scheme can withstand the typical attacks launched by external adversaries, for instance, the replay attack, man-in-the-middle attack, the impersonation attack, and the modification attack.

\item Key escrow freeness: Except for the vehicle itself, no one (even the LEA) can learn the entity's private key. 

\end{itemize}

Since communication in IoV has to fulfill the real-time demand, 
in addition to meeting the above security requirements, the proposed scheme has low computation and communication overheads.

\begin{figure}[!t]
\centering
\includegraphics[width=3.2in]{"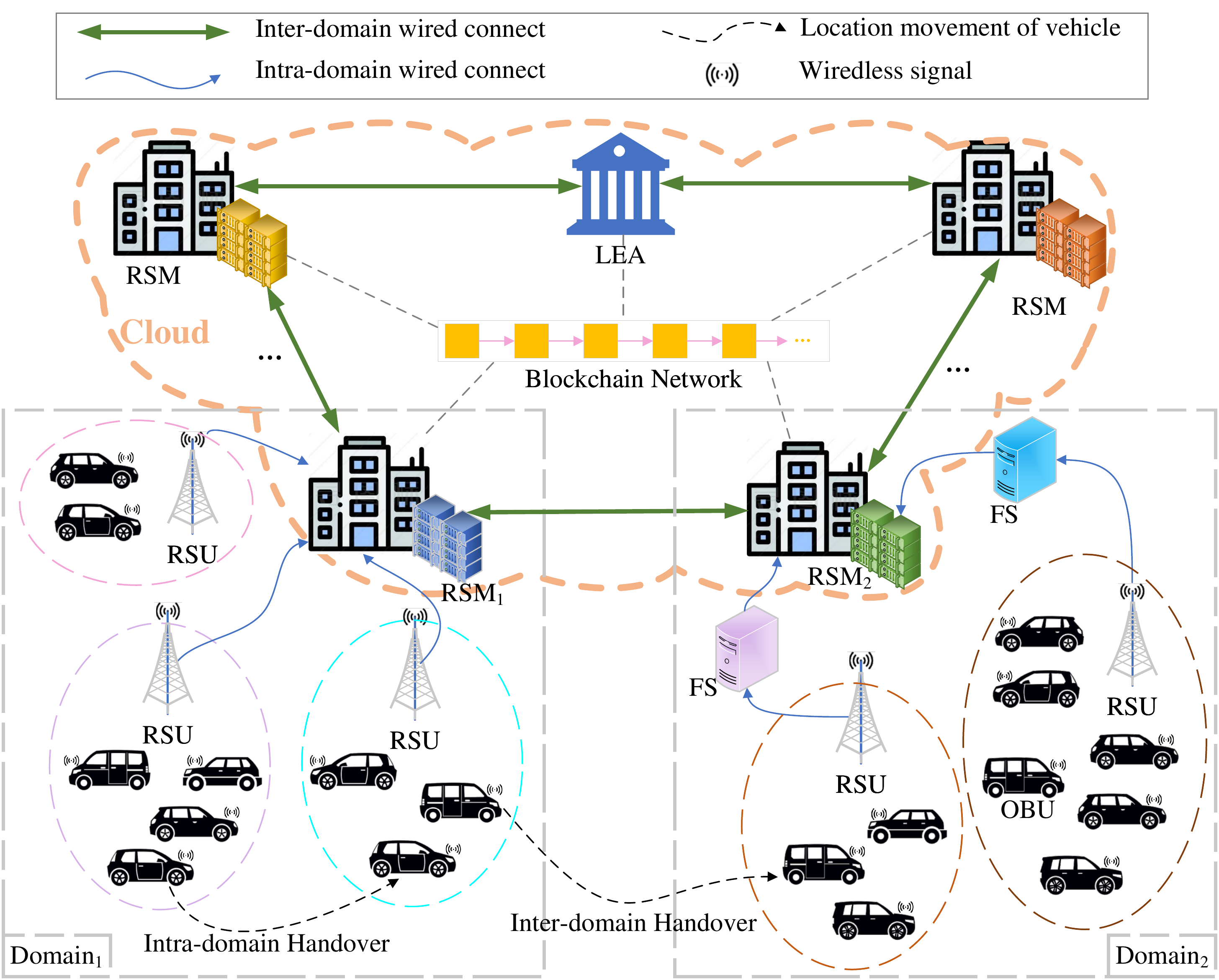"}
\caption{System Architecture.}
\label{Scenes schematic}
\end{figure}

\section{Proposed Scheme\label{Proposed Scheme}}
Necessary notations are summarized in Table \ref{notations} for ease of reference. BEPHAP consists of four phases: system initialization, network registration, handover authentication, and revocation.

\subsection{System Initialization}

This phase initializes the system parameters for each registered network member can be mathematically modeled as follows.

\subsubsection{LEA's Initialization}

\begin{itemize}

\item Let $t$ be a prime number, the finite field $\mathbb{F}_t$ is determined by the prime number $p$. Let $E(\mathbb{F}_t)$ be an elliptic curve over the finite field $\mathbb{F}_t$ and $P$ a point of $E(\mathbb{F}_t)$ with prime order $q$, $\mathbb{G}$ an additive elliptic curve group generated by $P$. Let $\lambda$ be a security parameter. The system initialization process is as follows:

\item The LEA choose hash functions: 


\begin{itemize}

    \item $H_0:\{0,1\}^* \times \mathbb{Z}_q^* \to \mathbb{Z}_q^*$,

    \item $H_1:\{0,1\}^{*} \times \mathbb{Z}_q^* \times \mathbb{Z}_q^* \times \{0,1\}^{*} \rightarrow \{0,1\}^{\lambda}$,
    
    \item $H_2:\{0,1\}^{*} \times \mathbb{Z}_q^* \times \mathbb{G} \times \{0,1\}^* \times \{0,1\}^{\lambda} \times \{0,1\}^* \to \mathbb{Z}_q^*$,
    
    \item $H_3:\mathbb{G} \times \{0,1\}^{\lambda} \times \mathbb{Z}_q^* \times \{0,1\}^* \to \{0,1\}^{\lambda}$,
    
    \item $H_4:\mathbb{Z}_q^* \times \{0,1\}^{\lambda} \times \{0,1\}^* \to \{0,1\}^{\lambda}$,
    
    \item $H_5:\{0,1\}^{\lambda+*} \times \mathbb{Z}_q^* \times \{0,1\}^{3\lambda+*} \to \{0,1\}^{\lambda}$.
    
    \item $H_6:\{0,1\}^{2\lambda+*} \times \mathbb{Z}_q^* \times \mathbb{G} \times \{0,1\}^{2\lambda+*} \to \{0,1\}^{\lambda}$

\end{itemize}

\item The LEA specifies a chameleon hash function denoted by $CH$ to be used by VNs.

\item The LEA generates a signing and verification key pair $(sk_{sig}^{LEA}, pk_{ve}^{LEA})$, and a decryption and encryption key pair $(sk_{de}^{LEA}, pk_{en}^{LEA})$ under an Elliptic Curve Digital Signature Algorithm (ECDSA). The signing and verification key pair can be used to send transactions to the blockchain.

\item Furthermore, the LEA randomly choose group key $GK, b \in \mathbb{Z}_q^*$, and send them to each RSM through secure channel. $GK, b$ need to be updated periodically.

\item Finally, the LEA publishes the system public parameter $para=\{q, P, \mathbb{G}, H_0, H_1, H_2, H_3, H_4, H_5\}$ and $pk_{ve}^{LEA}$, $pk_{en}^{LEA}$, then secretly saves $sk_{sig}^{LEA}, sk_{de}^{LEA}, GK, b$.

\end{itemize}

\subsubsection{RSM's initialization}

$RSM_x$ in domain $DM_y$ is initialized as follows:

\begin{itemize}
    \item The $RSM_x$ preloads with the system parameter $para$ and $pk_{ve}^{LEA}, pk_{en}^{LEA}$.
    
    \item The $RSM_x$ generates a signing and verification key pair $(sk_{sig}^{RSM_x}, pk_{ve}^{RSM_x})$, and a decryption and encryption key pair $(sk_{de}^{RSM_x}, pk_{en}^{RSM_x})$ under the ECDSA. Then the $RSM_x$ sends $\{pk_{ve}^{RSM_x}),pk_{en}^{RSM_x}\}$ to the LEA through the secure channel for registration.


    
    \item $LEA$ sends $GK,b$ to $RSM_x$ over the secure channel.
    
    \item $RSM_x$ receives and reserves $GK,b$.

    \item The RSM broadcasts $pk_{ve}^{RSM_x}, pk_{en}^{RSM_x}$.
    
\end{itemize}

\subsubsection{RSU's initialization}

$RSU_z$ subordinate to $RSM_x$ in domain $DM_y$ is initialized as follows:

\begin{itemize}
    \item The $RSU_z$ preloads with the system parameter $para$ and $pk_{ve}^{LEA}, pk_{en}^{LEA}$.
    
    \item The $RSU_z$ generates a signing and verification key pair $(sk_{sig}^{RSU_z}, pk_{ve}^{RSU_z})$, and a decryption and encryption key pair $(sk_{de}^{RSU_z}, pk_{en}^{RSU_z})$ under the ECDSA. Then the $RSU_z$ sends $\{pk_{ve}^{RSU_z},pk_{en}^{RSU_z}\}$ to the $RSM_x$ through the secure channel for registration.
    
    \item $RSM_x$ sends $GK,b$ to $RSU_z$ over the secure channel.
    
    \item $RSU_z$ receives and reserves $GK,b$
    
    \item The RSU broadcasts $pk_{ve}^{RSU}, pk_{en}^{RSU}$ and its identification $ID_{RSU}$.
    
\end{itemize}

\begin{table}[H]
	\renewcommand{\arraystretch}{1.3}
	\setlength{\abovecaptionskip}{0cm}
        \caption{Notations Used in SCHEME}
	\label{notations}
	\centering
	\begin{tabular}{cp{0.63\columnwidth}}
   \toprule
   Notation & Meaning \\
   \midrule
   $SEN_x(y)$ & Using a symmetric encryption function to encrypt plaintext $y$ with key $x$ \\
   $SDE_x(y)$ & Using a symmetric decryption function to decrypt ciphertext $y$ with key $x$ \\
   $AEN_x(y)$ & Encrypt plaintext $y$ using an asymmetric encryption function with key $x$ \\
   $ADE_x(y)$ & Decrypt ciphertext $y$ using an asymmetric decryption function with key $x$ \\
   $Sign_x(y)$ & Sign message $y$ with key $x$ \\
   $Veri_x(\sigma,y)$ & Verify the signature $\sigma$ with message $y$ and key $x$ \\
   
   \bottomrule
\end{tabular}
\end{table}

\subsection{Network Registration}

\begin{figure}[!t]
\centering
\includegraphics[width=3.4in]{"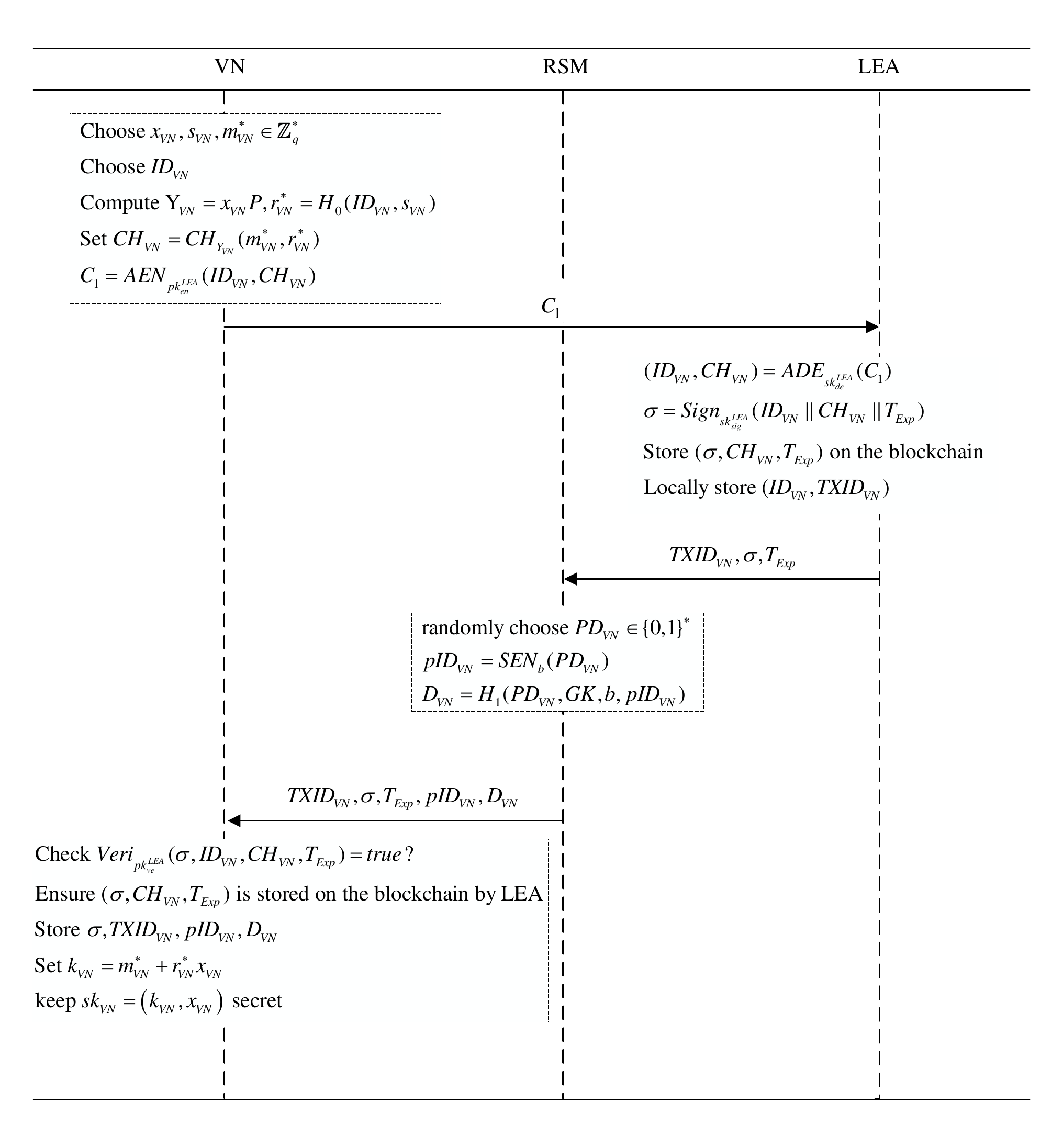"}
\caption{Network Registration.}
\label{registration}
\end{figure}

In order to complete the network registration, as shown in Fig. \ref{registration}, the following procedures are performed among the VN, the $RSM_x$, and the LEA. 

\begin{enumerate}
    \item The VN first chooses $x_{VN}, s_{VN}, m_{VN}^* \in \mathbb{Z}_q^*$ and the real identity $ID_{VN}$. Then the VN computes $Y_{VN}=x_{VN}P$, $r^*_{VN} = H_0(ID_{VN},s_{VN})$ and $CH_{VN}=CH_{Y_{VN}}(m_{VN}^*, r_{VN}^*)$, and sends message $C_1$ to $RSM_x$, where $C_1=AEN_{pk_{en}^{LEA}}(ID_{VN}, CH_{VN})$.
    
    \item  The $RSM_x$ forwards $C_1$ to the LEA. Upon receiving $C_1$, the LEA decrypts $C_1$ through $(ID_{VN},CH_{VN})=ADE_{sk_{de}^{LEA}}(C_1)$ and generates a signature $\sigma=Sign_{sk_{sig}^{LEA}}(ID_{VN}||CH_{VN}||T_{Exp})$, where $T_{Exp}$ is the expiration time of this registration. Then LEA stores $(\sigma,CH_{VN},T_{Exp})$ on the blockchain by sending a transaction. Let $TXID_{VN}$ be the transaction identity. The LEA locally stores $(ID_{VN},TXID_{VN})$ and sends $(TXID_{VN}, \sigma, T_{Exp})$ to the $RSM_x$.
    
    \item After receiving $(TXID_{VN}, \sigma, T_{Exp})$ from the LEA, the $RSM_x$ randomly chooses $PD_{VN} \in \{0,1\}^*$, and computes $pID_{VN}=SEN_b(PD_{VN})$, $D_{VN}=H_1(PD_{VN}, GK, b, pID_{VN})$. Then the $RSM_x$ sends $(TXID_{VN}, \sigma, T_{Exp}, pID_{VN}, D_{VN})$ to the VN.
    
    \item Upon receiving $(TXID_{VN}, \sigma, T_{Exp}, pID_{VN}, D_{VN})$, the VN verifies the signature $\sigma$ by $Veri_{pk_{ve}^{LEA}}(\sigma, ID_{VN}, CH_{VN}, T_{Exp})$ and ensures that $(\sigma, CH_{VN}, T_{Exp})$ is stored on the blockchain by the LEA.
    
    \item The VN locally stores $\sigma, TXID_{VN}, pID_{VN}, D_{VN}$, and keeps $sk_{VN}=(k_{VN},x_{VN})$ secret, where $k_{VN}=m_{VN}^*+r_{VN}^*x_{VN}$.
    
\end{enumerate}

\subsection{Cross-Domain Handover Authentication \label{Cross-Domain_Handover_Authentication}}

\begin{figure}[!t]
\centering
\includegraphics[width=3.4in]{"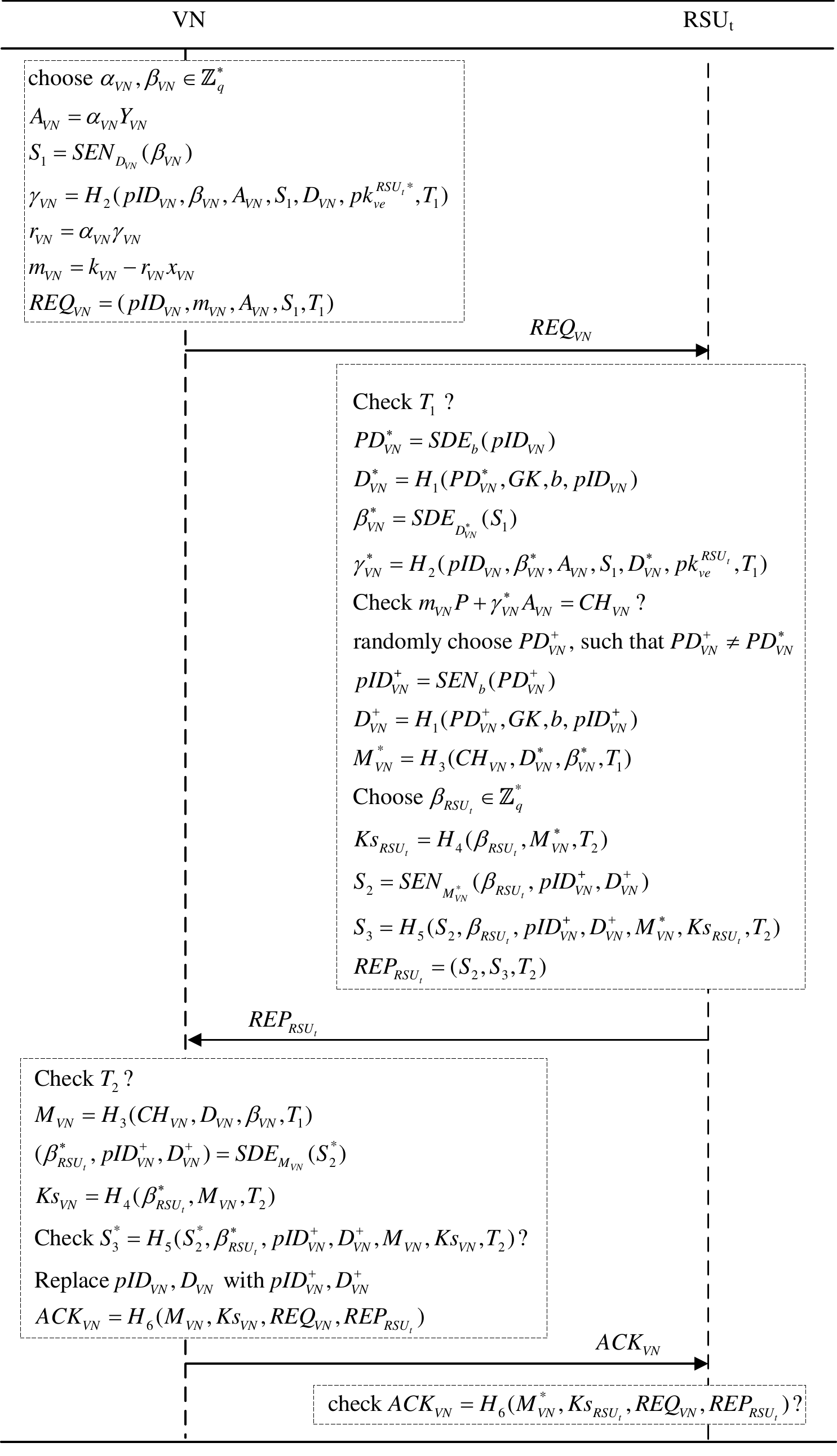"}
\caption{Cross-domain handover authentication and key agreement protocol.}
\label{authentication protocol}
\end{figure}

Referring to Fig. \ref{authentication protocol}, when a VN enters an area covered by $RSU_t$ belonging to $RSM_x$, a mutual authentication with key agreement between the VN and $RSU_t$ proceeds as follows.

\begin{enumerate}
    \item The VN first chooses $\alpha_{VN}, \beta_{VN} \in \mathbb{Z}_q^*$, and computes $A_{VN}=\alpha Y_{VN}$. The VN can complete the calculation of $A_{VN}$ before the authentication phase by pre-computation, so this step won't cause overhead in the authentication phase. Then the VN generates $S_1$, where $S_1=SEN_{D_{VN}}(\beta_{VN})$. To prove the validity of its pseudo-identity $pID_{VN}$, the VN calculate $\gamma_{VN}=H_2(PID_{VN},\beta_{VN},A_{VN},S_1,D_{VN},{pk_{ve}^{RSU_t}}^*,T_1)$, and computes $r_{VN}=\alpha_{VN} \gamma_{VN}$, $m_{VN}=k_{VN}-r_{VN}x_{VN}$, where $T_1$ is a timestamp. Since the RSUs are connected by wire and can communicate with each other, the VN can obtain the public keys of several RSUs along the way through the broadcast of a certain RSU before the authentication stage. Finally, the VN sends $REQ_{VN}=(pID_{VN}$, $m_{VN}$, $A_{VN}$, $S_1$, $T_1)$ to $RSU_t$.
    
    \item After receiving $REQ_{VN}$ from VN, $RSU_t$ check $T_1$'s validity to prevent replay attacks and calculate $PD^*_{VN}=SDE_b(pID_{VN})$, $D_{VN}^*=$ $H_1(PD_{VN}^*$, $GK$, $b$, $pID_{VN})$. Then $RSU_t$ sets $\gamma^*_{VN}=$ $H_2(pID_{VN}$, $\beta^*_{VN}$, $A_{VN}$, $S_1$, $D^*_{VN}$, $pk_{ve}^{RSU_t}$, $T_1)$ and checks the legitimacy of the VN based on Equation \eqref{chameleon}.
    
    \begin{equation}\label{chameleon}
        \ m_{VN}P+\gamma^*_{VN}A_{VN}=CH_{VN}
    \end{equation}
    
    Note that $CH_{Y_{VN}}(m_{VN}, r_{vN})=m_{VN}P+\gamma_{VN}A_{VN}$, where $r_{VN}=\alpha_{VN}\gamma_{VN}$. If Equation \eqref{chameleon} doesn't hold, the VN is invalid and $RSU_t$ quits. Otherwise, $RSU_t$ randomly chooses $PD^+_{VN}$, such that $PD^+_{VN}\neq PD^*_{VN}$, and computes $pID^+_{VN}=SEN_b(PD^+_{VN})$, $D^+_{VN}=H_1(PD^+_{VN}, GK, b, pID^+_{VN})$. $D^+_{VN}$ is the new symmetric key, which is used for symmetric encryption in the next round of cross-domain handover authentication. To make a key agreement with the VN, $RSU_t$ calculates a secret $M^*_{VN}=H_3(CH_{VN}, D^*_{VN}, \beta^*_{VN}, T_1)$ and chooses $\beta_{RSU_t} \in \mathbb{Z}^*_{q}$. The pairwise transient key $Ks_{RSU_t}$ is computed based on Equation \eqref{RSU session key}. 
    
    \begin{equation}\label{RSU session key}
        \ Ks_{RSU_t} = H_4(\beta_{RSU_t}, M^*_{VN}, T_2)
    \end{equation}
    
    Then $RSU_t$ chooses its own timestamp $T_2$ and calculates $S2=SEN_{M^*_{VN}}(\beta_{RSU_t}, pID^+_{VN}, D^+_{VN})$, $S_3=H_5(S_2, \beta_{RSU_t}, pID^+_{VN}, D^+_{VN}, M^*_{VN}, Ks_{RSU_t}, T_2)$. Finally, $RSU_t$ sends $REP_{RSU_t}$ to the VN, where $REP_{RSU_t}=(S_2, S_3, T_2)$.

    \item To thwart replay attacks, the VN first checks the freshness of $T_2$ after receiving $REP_{RSU_t}$ from $RSU_t$ and computes a secret $M_{VN}=H_3(CH_{VN}, D_{VN}, \beta_{VN}, T_1)$. Then the VN calculates $(\beta^*_{RSU_t}, pID^+_{VN}, D^+_{VN})=SDE_{M_{VN}}(S^*_2)$. The pairwise transient key $Ks_{VN}$ is computes based on Equation \eqref{VN session key}.

    \begin{equation}\label{VN session key}
        \ Ks_{VN} = H_4(\beta^*_{RSU_t}, M_{VN}, T_2)
    \end{equation}
    
    Finally, the VN checks the validity of the pairwise transient key agreement based on $S_3^*=H_5(S_2^*, \beta^*_{RSU_t}, pID^+_{VN}, D^+_{VN}, M_{VN}, Ks_{VN}, T_2)$. If successful, the VN replaces $pID_{VN}, D_{VN}$ with $pID^+_{VN}, D^+_{VN}$. then VN sends $ACK_{VN} = H_6(M_{VN}, Ks_{VN}, REQ_{VN}, REP_{RSU_t})$ to the $RSU_t$.

    \item After getting $ACK_{VN}$ from the VN, $RSU_t$ verifies the legitimacy of the pairwise transient key $Ks_{RSU_t}$ according to $ACK_{VN}= H_6(M_{VN}^*, Ks_{RSU_t}, REQ_{VN}, REP_{RSU_t})$
    
\end{enumerate}

The VN and $RSU_t$ can successfully conduct mutual authentication and key agreement through the above process and then communicate via the session key $Ks$.

\subsection{Revocation}
We use the revocation mechanism to punish hostile vehicles. After $RSM_x$ discovers a malicious vehicle $VN_m$, it will commit $CH_{VN_m}$ to the revocation list in the blockchain. Therefore, other RSMs can recognize $VN_m$ is malicious by reading the revocation list. In many related works, RSU prevents malicious vehicles from interacting with entities in the network by checking revocation lists \cite{zhu2013efficient}. While in BEPHAP, we can avoid malicious vehicles from accessing the network by updating the group key $(GK, b)$. The periodic update procedure is as follows:
\begin{itemize}
    \item LEA randomly chooses the $GK^{'}, b^{'} \in \mathbb{Z}_q^*$, and sends $(GK^{'}, b^{'})$ to all the RSMs over the secure channel.
    
    \item RSMs forwards $(GK^{'}, b^{'})$ to RSUs.
    
    \item For each honest vehicle, the RSUs use the new group key $(GK^{'}, b^{'})$ computes the new pseudo-identity $pID$ and new symmetric key $D$ for the honest vehicle. Specifically, for the honest vehicle $VN_h$, the $RSU_t$ randomly choose $PD_{VN_h}^{'}$, and computes $pID_{VN_h}^{'} = SEN_{b^{'}}(PD_{VN_h}^{'})$, $D_{VN_h}^{'}=H_1(PD_{VN_h}^{'}, GK^{'}, b^{'}, pID_{VN_h}^{'})$. Then the RSU sends $S_{upd}=SEN_{Ks_{RSU_t}}(pID_{VN_h}^{'}, D_{VN_h}^{'})$ to the VN.
    
    \item $VN_h$ receives $S_{upd}$, and calculate $SDE_{Ks_{VN_h}}(S_{upd})$ to obtain $pID_{VN_h}^{'}$ and $D_{VN_h}^{'}$.
    
    \item $VN_h$ replace $pID_{VN_h}, D_{VN_h}$ with $pID_{VN_h}^{'}$ and $D_{VN_h}^{'}$.
    
    \item Those malicious vehicles who haven't received new pseudo-identity $pID_{VN}^{'}$ and symmetric key $D_{VN}^{'}$ can't pass the message verification since the group key $(GK, b)$ have been updated to $(GK^{'}, b^{'})$.
\end{itemize}

Therefore, BEPHAP only needs to check the revocation list when the group key is updated, unlike other schemes where RSU checks the revocation list every time verifying the vehicle's authentication request.

\section{Security Evaluation\label{Security Evaluation}}
In this section, we formally prove the properties of mutual authentication and key agreement of BEPHAP based on the widely known Burrows-Abadi-Needham (BAN) logic \cite{DBLP:conf/sosp/BurrowsAN89}, which has been widely used to prove these two fundamental security properties of security protocols. Additionally, we verify various security properties of BEPHAP based on ProVerif tool and extensive analyses.

\subsection{Formal Security Proof Based on the BAN Logic}
Table \ref{BAN logic notations} and Table \ref{BAN logic rules} show the notations and rules of BAN logic, respectively. According to the BAN logic analytic procedure, we present the goals and the assumptions of BEPHAP, and we prove BEPHAP reaches these goals.

\begin{table}[H]
	\renewcommand{\arraystretch}{1.3}
	\setlength{\abovecaptionskip}{0cm}
        \caption{BAN logic Notations}
	\label{BAN logic notations}
	\centering
	\begin{tabular}{cc}
	\toprule
    Notation & Description\\
    \midrule
    
    $P \believes X$  & The entity $P$ believes the formula $X$ is true. \\
    $P \sees X$  & $P$ receives $X$. \\
    $P \oncesaid X$ & $P$ has once said $X$. \\
    $P \controls X$ &  $P$ has jurisdiction over $X$. \\
    $\fresh{X}$ & $X$ is fresh. \\
    $P \secret{X} Q$ & The entities $P$ and $Q$ share a secret $X$. \\
    $P \sharekey{K} Q$ &  The entities $P$ and $Q$ share a secret key $K$.\\
    $\encrypt{X}{K}$ & $X$ is encrypted based on the secret $K$.\\
    
    \bottomrule
\end{tabular}
\end{table}

\begin{table}[H]
	\renewcommand{\arraystretch}{2}
	\setlength{\abovecaptionskip}{0cm}
        \caption{BAN logic Rules}
	\label{BAN logic rules}
	\centering
	\begin{tabular}{p{0.4\columnwidth}c}
	\toprule
    Rule & Meaning\\
    \midrule
    
    $\frac{P \believes Q \sharekey{K} P, P \sees \encrypt{X}{K} }{ P \believes Q \oncesaid X }$, $\frac{P \believes Q \secret{Y} P, P \sees \encrypt{X}{Y} }{ P \believes Q \oncesaid X }  $  & \multirow{2}{*}{The message-meaning rules.} \\
    $\frac{P \believes \fresh{X}, P \believes Q \oncesaid X}{P \believes Q \believes X}$ & The nonce-verification rule. \\
    $\frac{P \sees (X,Y)}{P \sees X}$ & The seeing rule. \\
    $\frac{P \believes Q \controls X, P \believes Q \believes X}{P \believes X}$ &  The jurisdiction rule. \\
    $\frac{P \believes \fresh{X}}{P \believes \fresh{(X,Y)}}$ & The fresh-promotion rule. \\
    $\frac{P \believes X, P \believes Y}{P \believes (X,Y)}$ & The composition rule. \\
    $\frac{P \believes Q \believes (X,Y)}{P \believes Q \believes X}$, $\frac{P \believes (X,Y)}{P \believes X}$ & The decomposition rule.\\
    
    \bottomrule
\end{tabular}
\end{table}

\subsubsection{The Goals}
To achieve mutual authentication with key agreement between VN and $RSU_t$ is the fundamental goal of BEPHAP. Specifically, each entity not only has to believe the pairwise transient key $Ks$, but also needs to believe that the key is believed by the other entity. The goals of BEPHAP are as follows:

Goal 1. $VN \believes VN \sharekey{Ks} RSU_t$.

Goal 2. $RSU_t \believes RSU_t \sharekey{Ks} VN$.

Goal 3. $VN \believes RSU_t \believes RSU_t \sharekey{Ks} VN$.

Goal 4. $RSU_t \believes VN \believes VN \sharekey{Ks} RSU_t$.

\subsubsection{Assumptions}
There are several reasonable and necessary assumptions, since both $CH_{VN}$ and $D_{VN}$ are initially stored in the VN, and $CH_{VN}$ is initially stored in the $RSU_t$.

Assumption 1. $VN \believes CH_{VN}$.

Assumption 2. $RSU_t \believes CH_{VN}$.

Assumption 3. $VN \believes D_{VN}$.

\subsubsection{Security Result}
Theorem \ref{BAN-basis} gives security result of BEPHAP.

\begin{theorem}
\label{BAN-basis}
In BEPHAP, on the premise of ensuring the anonymity of the VN, the VN and the $RSU_t$ mutually authenticate each other and share a session key secretly.
\end{theorem}

\begin{proof}
First, we list the messages during the cross-domain handover authentication in BEPHAP. Then we prove BEPHAP reaches mutual authentication and key agreement on the premise of ensuring VN's identity anonymity. The details are described as follows:

Message 1. $VN \rightarrow RSU_t: REQ_{VN}$, where
\begin{equation*}
\ REQ_{VN}=(pID_{VN}, m_{VN}, A_{VN}, S_1, T_1).
\end{equation*}

Message 2. $RSU_t \rightarrow VN: REP_{RSU_t}$, where
\begin{equation*}
\ REP_{RSU_t}=(S_2, S_3, T_2).
\end{equation*}

Message 3. $VN \rightarrow RSU_t: ACK_{VN}$, where
\begin{equation*}
\ ACK_{VN} = H_6(M_{VN}, Ks_{VN}, REQ_{VN}, REP_{RSU_t})
\end{equation*}

According to Message 1, we have

\begin{steps}
  \item $RSU_t \sees REQ_{VN}$
  
  $RSU_t$ checks $T_1$ to thwart replay attacks. Thus,
  
  \item $RSU_t \believes \fresh{REQ_{VN}}$
  
  $RSU_t$ computes $PD^*_{VN}=SDE_b(pID_{VN})$, $D_{VN}^*=$ $H_1(PD_{VN}^*$, $GK$, $b$, $pID_{VN})$, $\beta_{VN}^* = SDE_{D_{VN}^*}(S_1)$. Then $RSU_t$ sets $\gamma^*_{VN}=$ $H_2(pID_{VN}$, $\beta^*_{VN}$, $A_{VN}$, $S_1$, $D^*_{VN}$, $pk_{ve}^{RSU_t}$, $T_1)$ and checks if $m_{VN}P+\gamma_{VN}^*A_{VN}=CH_{VN}$. 
  If it is, according to Step 1, we have:
  
  \item $RSU_t \believes VN \oncesaid REQ_{VN}$, $RSU_t \believes VN \controls REQ_{VN}$
  
  Based on Step 2, Step 3 and the nonce-verification rule, we have:
  
  \item $RSU_t \believes VN \believes REQ_{VN}$
  
  According to Step 4, and the decomposition rule, we have:
  
  \item $RSU_t \believes VN \believes \beta_{VN}^*$, $RSU_t \believes VN \believes D_{VN}^*$, $RSU_t \believes VN \believes T_1$
  
  Based on Step 5, Step 3 and the jurisdiction rule, we can get:
  
  \item $RSU_t \believes \beta_{VN}^*$, $RSU_t \believes D_{VN}^*$, $RSU_t \believes T_1$
  
  $RSU_t$ computes $M_{VN}^*=H_3(CH_{VN}, D_{VN}^*, \beta_{VN}^*, T_1)$. According to Step 6, and Assumption 2, we can get:
  
  \item $RSU_t \believes M_{VN}^*$
  
  $RSU_t$ chooses $\beta_{RSU_t}$ and its own timestamp $T_2$. Hence,
  
  \item $RSU_t \believes \beta_{RSU_t}$, $RSU_t \believes T_2$
  
  $RSU_t$ calculates $Ks_{RSU_t} = H_4(\beta_{RSU_t}, M_{VN}^*, T_2)$. Based on Step 7, Step 8, and composition rule, we can get:
  
  \item $RSU_t \believes Ks_{RSU_t}$

  That is $RSU_t \believes RSU_t \sharekey{Ks} VN \hfill \textit{(Goal 2)}$
  
  Based on Message 2, we have:
  
  \item $VN \sees (S_2, S_3, T_2)$
  
  The $VN$ checks $T_2$ to thwart replay attacks. Thus,
  
  
  
  \item $VN \believes \fresh{S_2^*}$, $VN \believes \fresh{S_3^*}$
  
  The $VN$ computes $M_{VN}=H_3(CH_{VN}, D_{VN}, \beta_{VN}, T_1)$. Since both $\beta_{VN}$ and $T_1$ are choosed by the $VN$, according to The Assumption 1 and Assumption 3, we have:
  
  \item $VN \believes M_{VN}$

  The $VN$ calculates $(\beta_{RSU_t}^*, pID^+_{VN}, D^+_{VN})=SDE_{M_{VN}}(S^*_2)$. Then the $VN$ computes $Ks_{VN}=H_4(\beta_{RSU_t}^*, M_{VN}, T_2)$, and checks if $S_3^* = H_5(S_2^*, \beta_{RSU_t}^*, pID^+_{VN}, D^+_{VN}, M_{VN}, Ks_{VN}, T_2)$. If it is, we have:
  
  \item $VN \believes RSU_t \controls REP_{RSU_t}$, $VN \believes \beta_{RSU_t}^*$, $VN \believes T_2$, $VN \believes Ks_{VN}$,
  
  Note that
  \[
  \ Ks_{VN} = H_4(\beta_{RSU_t}^*, M_{VN}, T_2) = Ks_{RSU_t}
  \]
  
  Hence,
  \item $VN \believes VN \sharekey{Ks} RSU_t \hfill \textit{(Goal 1)}$
  
  Note that $M_{VN} = M_{VN}^* \triangleq M$, according to Step 12, we have:
  
  \item $VN \believes VN \secret{M} RSU_t$
  
  According to Message 2, Step 10, Step 15, and message-meaning rule, we have:
  
  \item $VN \believes RSU_t \oncesaid Ks_{RSU_t}$
  
  According to Step 11, and fresh-promotion rule, we have:
  
  \item $VN \believes \fresh{Ks_{RSU_t}}$
  
  Based on Step 16, Step 17, and nonce-verification rule, we have:
  
  \item $VN \believes RSU_t \believes Ks_{RSU_t}$
  
  Note that $Ks_{RSU_t} = Ks_{VN}$, we can get:
  
  \item $VN \believes RSU_t \believes RSU_t \sharekey{Ks} VN \hfill \textit{(Goal 3)}$
  
  According to Message 3, we obtain:
  
  \item $RSU_t \sees ACK_{VN}$
  
  Since $M_{VN}^* = M_{VN} \triangleq M$, according to Step 7, we can get:
  
  \item $RSU_t \believes RSU_t \secret{M} VN$
  
  $RSU_t$ checks if $ACK_{VN} = H_6(M_{VN}^*,$ $Ks_{RSU_t},$ $REQ_{VN}, $ $REP_{RSU_t})$. if it is, according to Message 3, Step 20, Step 21 and the message-meaning rule, we have:
  
  \item $RSU_t \believes VN \oncesaid (Ks_{RSU_t}, REQ_{VN})$
  
  Based on Step 2, and the fresh-promotion rule, we have: 
  
  \item $RSU_t \believes \fresh{Ks_{RSU_t}, REQ_{VN}}$
  
  According to Step 22, Step 23, and the nonce-verification rule, we obtain:
  
  \item $RSU_t \believes VN \believes (Ks_{RSU_t}, REQ_{VN})$
  
  According to Step 24, and the decomposition rule, we have:
  
  \item $RSU_t \believes VN \believes Ks_{RSU_t}$
  
  Note that $Ks_{RSU_t} = Ks_{VN}$, we obtain:
  
  \item $RSU_t \believes VN \believes VN \sharekey{Ks} RSU_t \hfill \textit{(Goal 4)}$
\end{steps}

\qed
\end{proof}

To summarize, BEPHAP achieves the four goals. Both $VN$ and $RSU_t$ believe that they share $Ks$ with one another and believe the other is a legitimate entity. The identity of the $VN$ is hidden throughout the process, thus achieving anonymity.

\subsection{Formal Verification}
In this section, we model BEPHAP and check its security by ProVerif, which is a widely known cryptographic protocol verification tool in the formal model (so-called Dolev-Yao model \cite{DBLP:journals/tit/DolevY83}). The security of various authentication protocols or encryption schemes can be proved by ProVerif, for instance, Diffie-Hellman key exchange algorithms, hash functions, and signature schemes\cite{blanchet2018proverif}\cite{bauer2000security}\cite{DBLP:journals/jpdc/XuLLXJ21}. We define the following eight events:

\begin{itemize}
    \item \textit{event VNAcRSU}: The VN authenticated the RSU successfully.
    \item \textit{event RSUAcVN}: The RSU authenticated the VN successfully.
    \item \textit{event VNTerm}: The VN completed the authentication protocol.
    \item \textit{event RSUTerm}: The RSU completed the authentication protocol.
    \item \textit{event e1}: the first message is sent.
    \item \textit{event e2}: the second message is sent.
    \item \textit{event e3}: the third message is sent.
    \item \textit{event endmessage}: All messages have been sent and the protocol has completed.

\end{itemize}

We use ProVerif to verify that VN and RSU successfully perform mutual authentication and key agreement, with all messages sent in the correct order. Additionally, we verify the randomness of all messages and the strong secrecy of session key $Ks$, and secret value $M_{VN}, D_{VN}, D_{VN}^+, \beta_{VN}, \beta_{RSU}, pID_{VN}^+, \alpha_{VN}$, where the strong secrecy means that the attacker is incapable of distinguishing when the secret changes. Table \ref{Proverif result} shows the verification results of our protocol. The results indicate that our protocol guarantees the confidentiality of parameters $k,x,GK,b,PD$ and achieves mutual authentication and key agreement in which all events are executed in order. And the attacker can't obtain the secret values in the protocol. The strong secrecy and randomness are verified by observational equivalence, and the output in ProVerif is true, which means the verification is successful. Our source code is published on Github\footnote{https://github.com/KenHopkin/protocol-verification}.

\begin{table}
\centering
\caption{The Verification Results Given by The ProVerif}
\label{Proverif result}
\begin{threeparttable}
\resizebox{1.0\columnwidth}{!}{
\begin{tabular}{l}
\toprule
secrecy assumption verified: fact unreachable attacker(k), ok \\
secrecy assumption verified: fact unreachable attacker(x), ok \\
secrecy assumption verified: fact unreachable attacker(GK), ok \\
secrecy assumption verified: fact unreachable attacker(b), ok \\
secrecy assumption verified: fact unreachable attacker(PD), ok \\
inj-event(VNTerm()) $\Longrightarrow$ inj-event(RSUAcVN()) is true. \\
inj-event(RSUTerm()) $\Longrightarrow$ inj-event(VNAcRSU()) is true. \\
event(RSUTerm(VN,x,Ks)) $\&\&$ event(VNAcRSU(VN,x,Ks')) $\Longrightarrow$ Ks = Ks' is true. \\
event(VNTerm(x,RSU,Ks)) $\&\&$ event(RSUAcVN(x,RSU,Ks')) $\Longrightarrow$ Ks = Ks' is true. \\
event(RSUTerm(x,y,k)) $\&\&$ event(VNAcRSU(x',y',k)) $\Longrightarrow$ x = x' $\&\&$ y = y' is true. \\
event(VNTerm(x,y,k)) $\&\&$ event(RSUAcVN(x',y',k)) $\Longrightarrow$ x = x' $\&\&$ y = y' is true. \\
not attacker (M) is true. \\
not attacker ($\beta_{VN}$) is true. \\
not attacker ($\beta_{RSU}$) is true.\\
not attacker (Ks) is true.\\
not attacker ($D_{VN}$) is true.\\
not attacker ($D_{VN}^+$) is true.\\
not attacker ($pID_{VN}^+$) is true. \\
not attacker ($\alpha_{VN}$) is true. \\
inj-event(endmessage()) $\Longrightarrow$ (inj-event(e3()) $\Longrightarrow$ (inj-event(e2()) $\Longrightarrow$ inj-event(e1()))) is true. \\
Observational equivalence is true. \\
\bottomrule
\end{tabular}
}

\end{threeparttable}
  
\end{table}

\subsection{Further Security Analysis}

In this section, we further demonstrate that BEPHAP possesses multiple security properties. Since BEPHAP's resistance to the basic \textit{Man-in-the-middle}, \textit{Impersonation}, and \textit{Replay} attacks has been verified by ProVerif, we only show other important security properties in the following.

\subsubsection{Mutual authentication with key agreement}

for mutual authentication, $RSU_t$ checks the legitimacy of the $VN$ according to Equation \ref{chameleon}, whose soundness is presented as follows:

\begin{eqnarray}    \label{mutual_authentication_eq}
m_{VN}P+\gamma_{VN}^*A_{VN}&=&(k_{VN}-r_{VN}x_{VN})P + \gamma_{VN}^*A_{VN} \nonumber    \\
~&=&(k_{VN}-r_{VN}x_{VN})P + \gamma_{VN}^*\alpha_{VN}Y_{VN} \nonumber    \\
~&=&k_{VN}P \nonumber    \\
~&=&(m_{VN}^*+r_{VN}^*x_{VN})P \nonumber    \\
~&=&CH_{VN}
\end{eqnarray}

The $VN$ verifies the legitimacy of the $RSU_t$ based on $S_3^*=H_5(S_2^*, \beta^*_{RSU_t}, pID^+_{VN}, D^+_{VN}, M_{VN}, Ks_{VN}, T_2)$. If successful, it means that RSU owns $GK,b$.

As for key agreement, it's obvious that 
\[
  \ Ks_{VN} = H_4(\beta_{RSU_t}^*, M_{VN}, T_2) = Ks_{RSU_t}
\]

According to the protocol, $Ks = Ks_{VN} = Ks_{RSU_t}$ is the pairwise transient key.

\subsubsection{Data Confidentiality}

As claimed by the threat model defined in \ref{Threat Model}, the attacker can read and intercept all the messages sent over insecure channels. Hence, the attacker can obtain $pID_{VN}$, $m_{VN}$, $A_{VN}$, $S_1$ to $s_3$, $T_1$, $T_2$, $ACK_{VN}$. It has been verified by ProVerif that secret values including $k_{VN}, x_{VN}, \alpha_{VN}, Gk, b, pID_{VN}^+, D_{VN}^+, Ks$, are not exposed to the external adversaries. Then our protocol should ensure that $k_{VN}, x_{VN}, \alpha_{VN}$ are not exposed to internal adversaries. According to the ECDLP stated in \ref{Elliptic Curve Cryptosystem}, the internal adversary cannot get $\alpha_{VN}$ based on $P$ and $A_{VN}$. Obviously, the internal adversary also cannot obtain $k,x$ only based on $m_{VN}$.

\subsubsection{Identity Anonymity}
In the registration phase, VN encrypts $ID_{VN}$ with $pk_{en}^{LEA}$, and LEA store $(\sigma, CH_{VN}, T_{Exp})$ without leaking $ID_{VN}$. In the cross-domain handover authentication phase, $pID_{VN}$, the VN's pseudo-identity is adopted rather than the real identity $ID_{VN}$. Hence, no one except for LEA can reveal VN's real identity and the identity anonymity is achieved.

\subsubsection{Unlinkability}

The Proposed scheme achieves unlinkability that no attacker can link multiple messages to the same VN except for the entities that initially need this information, including LEA, RSM, and RSU. Since the $pID_{VN}$ is a random value encrypted by AES which is a pseudo-random function, according to \ref{pseudo-random function}, no adversary can tell whether two $pID_{VN}$ are derived from the same VN. And the randomness of $pID_{VN}$ has also been proved by ProVerif. What's more, the adversary cannot calculate $CH_{VN}$ due to the lack of $GK, b, D_{VN}$ for obtaining $\beta_{VN}$. So the adversary cannot link multiple messages based on $CH_{VN}$. Therefore, the unlinkability is achieved.

\subsubsection{Traceability}

Suppose the VN is detected to be malicious, RSU can generate a signature $\sigma_{RT}=Sign_{sk_{sig}^{RSU}}(REQ_{VN})$, and send $(\sigma_{RT}, REQ_{VN})$ to LEA. Then the LEA can trace the real identity of the malicious VN as follows:
\begin{itemize}
    \item Compute $PD_{VN}^* = SDE_b(pID_{VN})$, $D_{VN}^* = H_1(PD_{VN}^*, GK, b, pID_{VN})$,  $\beta_{VN}^* = SDE_{D_{VN}^*}(S_1)$.
    \item Calculate $\gamma_{VN}^*$ $=H_2(pID_{VN},$ $\beta_{VN}^*,$ $A_{VN},$ $S_1,$ $D_{VN}^*,$ $pk_{ve}^{RSU_t},$ $T_1)$.
    \item Set $CH_{VN}=m_{VN}P+\gamma_{VN}^*A_{VN}$.
    \item Find the transaction identity $TXID_{VN}$ in the blockchain based on $CH_{VN}$.
    \item According to $TXID_{VN}$, find the item $(ID_{VN}, TXID_{VN})$ from the local storage and output $(ID_{VN}, D_{VN}^*)$.
\end{itemize}

LEA can also obtain authentication records on $CH_{VN}$ from multiple RSUs to track the trajectory of the malicious VN.

\subsubsection{Non-repudiation}

As stated in \ref{Chameleon Hash Function}, the chameleon hash function has collision resistance. Only the owner of the private key $(k_{VN},x_{VN})$ can generate the collision of the chameleon hash. Since we use ECC in \ref{Elliptic Curve Cryptosystem} to implement it, if the attacker wants to find the collision, the attacker needs to solve the ECDLP, but this is not feasible. Hence, The proposed scheme achieves Non-repudiation.

\subsubsection{Non-frameability}

Since only the private key owner can generate the parameters for calculating $CH_{VN}$, no other entity can forge the message sent by the VN to frame the VN. However, LEA can trace malicious vehicles, and it may output the ID of an honest VN in tracing a malicious VN, thereby framing the honest one. Specifically, as for $(\sigma_{RT}, REQ_{VN})$, where $REQ_{VN_m} = pID_{VN_m}, m_{VN_m}, A_{VN_m}, S_1, T_1$ which is a malicious VN's request, assume $(ID_{VN_m}, D_{VN_m}^*)$ is the result that LEA gets after executing traceability. But, the LEA output $(ID_{VN_h}, D_{VN_m}^*)$, where $ID_{VN_h} \neq ID_{VN_m}$ to frame $VN_h$, which is an honest VN. In this situation, the $VN_h$ can give $TXID_{VN_h}$. Then, a third party can verify that the VN is framed by the LEA as follows:
\begin{itemize}
    \item Check $Veri_{pk_{ve}^{RSU}}(\sigma_{RT},REQ_{VN})=ture?$
    \item Compute $\gamma_{VN_m}^*=$ $H_2(pID_{VN_m},$ $\beta_{VN_m}^*,$ $A_{VN_m},$ $S_1,$ $D_{VN_m}^*,$ $pk_{ve}^{RSU}, T_1)$.
    \item Set $CH_{VN_m}=m_{VN_m}P+\gamma_{VN_m}A_{VN_m}$
    \item Obtain $(\sigma, CH_{VN_h}, T_{Exp})$ from the blockchain based on $TXID_{VN_h}$.
    \item Check $Veri_{pk_{ve}^{LEA}}(\sigma, ID_{VN_h} || CH_{VN_h} || T_{Exp})=true?$
    \item If it's true and $CH_{VN_h} \neq CH_{VN_m}$, the third party believes $VN_h$ is framed by the LEA.
    
\end{itemize}

\subsubsection{Key Escrow Freeness}
Since the VNs' secret keys are entirely chosen by themselves in the registration phase, BEPHAP is a key escrow freeness authentication protocol with key agreement.

\subsubsection{Cross-domain}
The VN in BEPHAP can securely perform handover authentication between different domains just like the intra-domain handover, based on the global synchronization of the blockchain, the trapdoor collision property of the chameleon hash function, and the security properties of ECC and pseudo-random functions.

\begin{table*}
\renewcommand\tabcolsep{3.5pt} 
\centering
\caption{Comparison of Functionality }
\label{functionality compaison}
\begin{threeparttable}

\begin{tabular}{lcccccccccccc}
\toprule
\diagbox{Scheme}{Functionality} & MA & DC & IA &  Unlinkability & Traceability & Non-repudiation & Non-frameability & KEF & Cross-domain & KA & FSP & FSV\\
\midrule
Zhang et al.\cite{DBLP:journals/tdsc/0002DB021} & Yes & Yes & Yes & No & Yes & Yes & Yes & Yes & Yes & Yes & Yes & Yes \\ 
Yang et al.\cite{DBLP:journals/tifs/YangZZCZ22} & No & Yes & Yes & Yes & Yes & Yes & Yes & Yes & – & No & Yes & No \\
Wei et al.\cite{DBLP:journals/tifs/WeiCXCZ21} & No & No & Yes & Yes & Yes & Yes & Yes & Yes & – & No & Yes & No \\
Feng et al.\cite{DBLP:journals/tifs/FengSXW21} & No & No & Yes & Yes & Yes & Yes & Yes & Yes & Yes & No & Yes & No \\
Jiang et al.\cite{DBLP:journals/tits/JiangZW16} & Yes & No & Yes & Yes & Yes & Yes & No & No & Yes & No & Yes & No \\
Wang et al.\cite{DBLP:journals/tits/WangCKSTL21} & Yes & No & Yes & Yes & Yes & Yes & No & No & – & No & Yes & No \\
Xu et al.\cite{DBLP:journals/jpdc/XuLLXJ21} & Yes & Yes & Yes & Yes & Yes & Yes & No & No & Yes & Yes & Yes & Yes\\
Feng et al.\cite{DBLP:journals/jsa/FengSXL21} & Yes & Yes & Yes & Yes & Yes & Yes & No & No & Yes & No & Yes & No\\
Lu et al.\cite{DBLP:journals/tvlsi/LuWQZL19} & No & No & Yes & No & Yes & Yes & Yes & Yes & – & No & Yes & No\\
BEPHAP & Yes & Yes & Yes & Yes & Yes & Yes & Yes & Yes & Yes & Yes & Yes & Yes\\
\bottomrule
\end{tabular}

\label{table:functionality notation}

    \begin{tablenotes}
        \footnotesize
        \item[1] \textbf{MA}: mutual authentication. \textbf{DC}: data confidentiality. \textbf{IA}: identity anonymity. \textbf{KEF}: key escrow freeness. \textbf{KA}: key agreement. \textbf{FSP}: formal security proof. \textbf{FSV}: Verification by formal security verification tools.
        
        \item[2] The symbol “–” means the functionality is not involved.

    \end{tablenotes}
\end{threeparttable}
  
\end{table*}

\section{Functionality and Performance Evaluation\label{Functionality and Performance Evaluation}}
In this section, we analyze the functionality and performance of BEPHAP and compare it with previous schemes. 

\subsection{Functionality Comparison}
We present the functionality comparison of BEPHAP and related approaches in Table \ref{functionality compaison}.

Note that \cite{DBLP:journals/tifs/YangZZCZ22}, \cite{DBLP:journals/tifs/WeiCXCZ21}, \cite{DBLP:journals/tifs/FengSXW21}, \cite{DBLP:journals/tvlsi/LuWQZL19} do not implement mutual authentication, which means that one of the two communicating parties may be impersonated by an attacker. \cite{DBLP:journals/tdsc/0002DB021}, \cite{DBLP:journals/tvlsi/LuWQZL19} do not provide the unlinkability, which may lead to the leakage of vehicle trajectory privacy. \cite{DBLP:journals/tits/JiangZW16}, \cite{DBLP:journals/tits/WangCKSTL21}, \cite{DBLP:journals/jpdc/XuLLXJ21}, \cite{DBLP:journals/jsa/FengSXL21} do not achieve non-frameability, which means that the user in their scheme may be framed by an external or internal attacker.

To the best of our knowledge, BEPHAP is the first authentication protocol scheme that simultaneously implements all properties in Table \ref{functionality compaison} in the field of IoV, including mutual authentication, data confidentiality, identity anonymity, unlinkability, traceability, non-repudiation, non-frameability, key escrow freeness, cross-domain, key agreement, formal security proof, and verification by formal security verification tools.

\begin{table*}
\renewcommand\tabcolsep{4.5pt} 
\centering
\caption{Comparison of Computation Overhead}
\label{computation compaison}
\begin{threeparttable}
\resizebox{2\columnwidth}{!}{
\begin{tabular}{l|c|c|c}
\toprule
Scheme & Single VN request & $RSU_t$ verify a single request & $RSU_t$ verify $n$ requests\\
\midrule
RUSH\cite{DBLP:journals/tdsc/0002DB021} & $4T_H + T_{SM-ECC} + T_{MSM-ECC}$ & $5T_H+ T_{SM-ECC} + T_{MSM-ECC}$ & $5nT_H+ n(T_{SM-ECC} + T_{MSM-ECC})$ \\


PPAAS\cite{DBLP:journals/tifs/YangZZCZ22} & $2T_H + 6T_{SM1-BP}$ & $2T_H + 5T_{BP} + 3T_{SM1-BP} + 3T_{SMT-BP} + 2T_{MTP}$ & $2T_H + 5T_{BP} + 3nT_{SM1-BP} + 3T_{SMT-BP} + 2nT_{MTP}$ \\


P2BA\cite{DBLP:journals/tifs/FengSXW21} & $2T_{BP}+11T_{SM1-BP}+12T_{EX-BP}$ & $4T_{BP}+10T_{SM1-BP}+10T_{EX-BP}$ & $4T_{BP}+(6n-1)T_{SM1-BP}$ \\


HDMA\cite{DBLP:journals/tits/WangCKSTL21} & $T_{EX} + T_{RSA-V} + T_{RSA-DE}$ & $T_{RSA-V} + T_{RSA-EN} + T_{RSA-DE}$ & $nT_{RSA-V} + nT_{RSA-EN} + nT_{RSA-DE}$ \\


BEPHAP & $T_{AES-EN} + 5T_{H} +T_{AES-DE}$ & $2T_{AES-DE} + 7T_{H} + T_{MSM-ECC} + 2T_{AES-EN} $ & $2nT_{AES-DE} + 7nT_{H} + nT_{MSM-ECC} + 2nT_{AES-EN}$ \\

\bottomrule
\end{tabular}
}

\end{threeparttable}
  
\end{table*}

\subsection{Computation Cost}


In this section, we compare BEPHAP with its related work. We use the standard cryptographic library MIRACL \cite{Miracl_Library}, a multi-precision integer and rational arithmetic C/C++ library, to simulate some operations' computation overhead with an Intel i5-6200U CPU @ 2.40 GHz as the VN and an Intel i9-10900 CPU @ 2.81 GHz as the $RSU_t$. The overhead of each operation is listed in Table \ref{cryptographic operations}. $T_{VN}$ denotes the execution time of the operations on the VN, and $T_{RSU}$ is the execution time of the operations on the $RSU_t$. Note that $T_{MSM-ECC} = 1.25T_{SM-ECC}$\cite{DBLP:journals/twc/YangHWD10}. 

The computation overhead of various schemes is listed in Table \ref{computation compaison}. In RUSH \cite{DBLP:journals/tdsc/0002DB021}, the computation cost on VN is $4T_H + T_{SM-ECC} + T_{MSM-ECC} = 0.83296$ ms for a single request, whereas the computation overhead on RSU is $5T_H+ T_{SM-ECC} + T_{MSM-ECC} = 0.2973$ ms for verifying a single request, and $5nT_H+ n(T_{SM-ECC} + T_{MSM-ECC})$ for verifying $n$ requests. Note that the computational overhead is optimized by pre-computation in\cite{DBLP:journals/tdsc/0002DB021}. In PPAAS \cite{DBLP:journals/tifs/YangZZCZ22}, the computation overhead is $2T_H + 6T_{SM1-BP} = 2.84978$ ms on VN for a single request, $2T_H + 5T_{BP} + 3T_{SM1-BP} + 3T_{SMT-BP} + 2T_{MTP} = 3.98364$ ms on RSU for verifying a single request, $2T_H + 5T_{BP} + 3nT_{SM1-BP} + 3T_{SMT-BP} + 2nT_{MTP}$ on RSU for verifying $n$ requests. In P2BA  \cite{DBLP:journals/tifs/FengSXW21}, the computation cost on VN is $2T_{BP}+11T_{SM1-BP}+12T_{EX-BP} = 21.6964$ ms for a single request, whereas the computation overhead on RSU is $4T_{BP}+10T_{SM1-BP}+10T_{EX-BP} = 8.38756$ ms for verifying a single request, and $4T_{BP}+(6n-1)T_{SM1-BP}$ for verifying $n$ requests. In HDMA \cite{DBLP:journals/tits/WangCKSTL21}, the computation overhead is $T_{EX} + T_{RSA-V} + T_{RSA-DE} = 11.11611$ ms on VN for a single request, $T_{RSA-V} + T_{RSA-EN} + T_{RSA-DE} = 4.60747$ ms on RSU for verifying a single request, $nT_{RSA-V} + nT_{RSA-EN} + nT_{RSA-DE}$ on RSU for verifying $n$ requests. As for BEPHAP, the computation cost is $T_{AES-EN} + 5T_{H} +T_{AES-DE} = 0.01$ ms on VN for a single request, $2T_{AES-DE} + 7T_{H} + T_{MSM-ECC} + 2T_{AES-EN} = 0.16856$ ms on RSU for verifying a single request, $2nT_{AES-DE} + 7nT_{H} + nT_{MSM-ECC} + 2nT_{AES-EN}$ on RSU for verifying $n$ requests. It is worth noting that the computation of $A_{VN}$ in \ref{Cross-Domain_Handover_Authentication} can be calculated before the authentication phase by pre-computation. Fig.\ref{Computation on RSU} illustrates RSU's computation versus the number
of requests. Obviously, the computation cost of
our is the lowest compared with other schemes on RSU. Compared with HDMA \cite{DBLP:journals/tits/WangCKSTL21}, BEPHAP reduces the computational cost of RSU by two orders of magnitude. The vertical axis adopts the log scale to make a more explicit comparison.  Fig.\ref{Computation on VN} shows the computational overhead on VN side. It can be observed that BEPHAP is much better than other schemes. It is worth noting that BEPHAP reduces the computational cost of VN by two or even three orders of magnitude compared to other schemes, which indicates that BEPHAP is much more suitable for IoV scenarios with limited vehicle computing capacity. Fig.\ref{average authentication delay} illustrates the average authentication delay versus the number of requests. As Fig.\ref{average authentication delay} shows, BEPHAP achieves the best performance in average authentication delay.

\begin{table}[H]
	\renewcommand{\arraystretch}{1.3}
	\renewcommand\tabcolsep{3.3pt} 
	\setlength{\abovecaptionskip}{0cm}
    \caption{Execution Time of Several Cryptographic Operations (ms)}
	\label{cryptographic operations}
	\centering
	
	\resizebox{1.0\columnwidth}{!}{
	\begin{tabular}{c|p{0.55\columnwidth}|c|c}
	\toprule
    Operation & Detail & $T_{VN}$ & $T_{RSU}$\\
    \midrule

    $T_{SM-ECC}$  & Scalar multiplication related to the ECC  & 0.36688 & 0.13060\\
    $T_{MSM-ECC}$  & Multi elliptic curve scalar multiplication & 0.45860 & 0.16325 \\
    $T_H$ & One-way hash function & 0.00187 & 0.00069 \\
    $T_{BP}$ & Bilinear paring & 1.38217 & 0.51329 \\
    $T_{SM1-BP}$ & Scalar multiplication operation in $\mathbb{G}_1$ related to bilinear pairing & 0.47434 & 0.12248 \\
    $T_{SMT-BP}$ & Scalar multiplication operation in $\mathbb{G}_T$ related to bilinear pairing & 0.24195 & 0.10825 \\
    $T_{EX-BP}$ & Exponentiation related to bilinear pairing & 1.14286 & 0.51096 \\
    $T_{MTP}$ & MapToPoint hash operation of the bilinear pairing & 0.74839 & 0.36181 \\
    $T_{EX}$ & Modular exponentiation & 2.77497 & 1.15168 \\
    $T_{RSA-G}$ & RSA signature generation & 7.88791 & 4.10632 \\
    $T_{RSA-V}$ & RSA signature verification & 0.52473 & 0.27127 \\
    $T_{RSA-EN}$ & RSA encryption & 0.50510 & 0.26081 \\
    $T_{RSA-DE}$ & RSA decryption & 7.81641 & 4.07539 \\
    $T_{AES-EN}$ & AES encryption & 0.00032 & 0.00012\\
    $T_{AES-DE}$ & AES decryption & 0.00033 & 0.00012\\
    
    \bottomrule
\end{tabular}
}
\end{table}

\begin{figure}[!t]
\centering
\includegraphics[width=2.5in]{"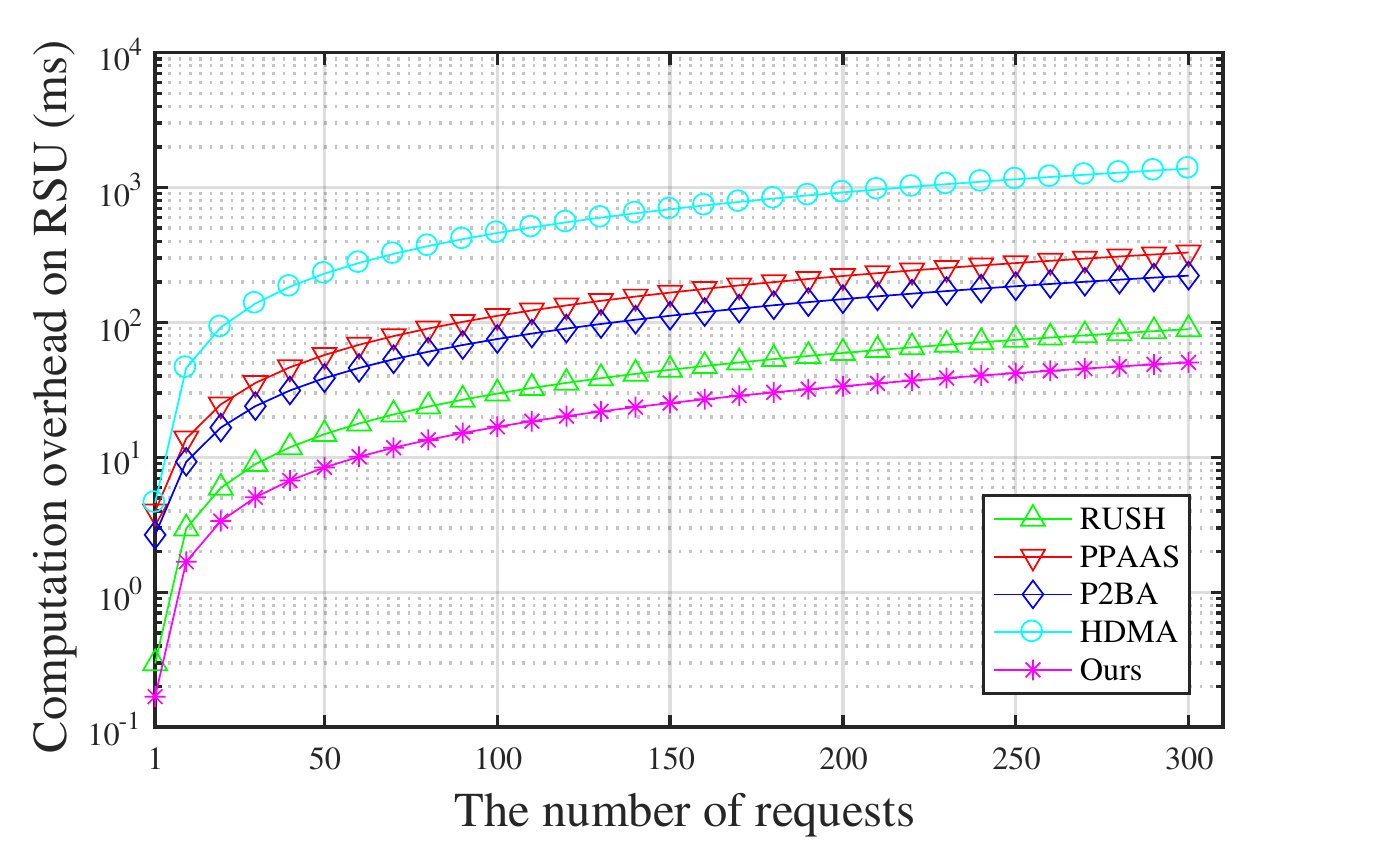"}
\caption{Computational overhead of multiple requests on RSU.}
\label{Computation on RSU}
\end{figure}

\begin{figure}[!t]
\centering
\includegraphics[width=2.5in]{"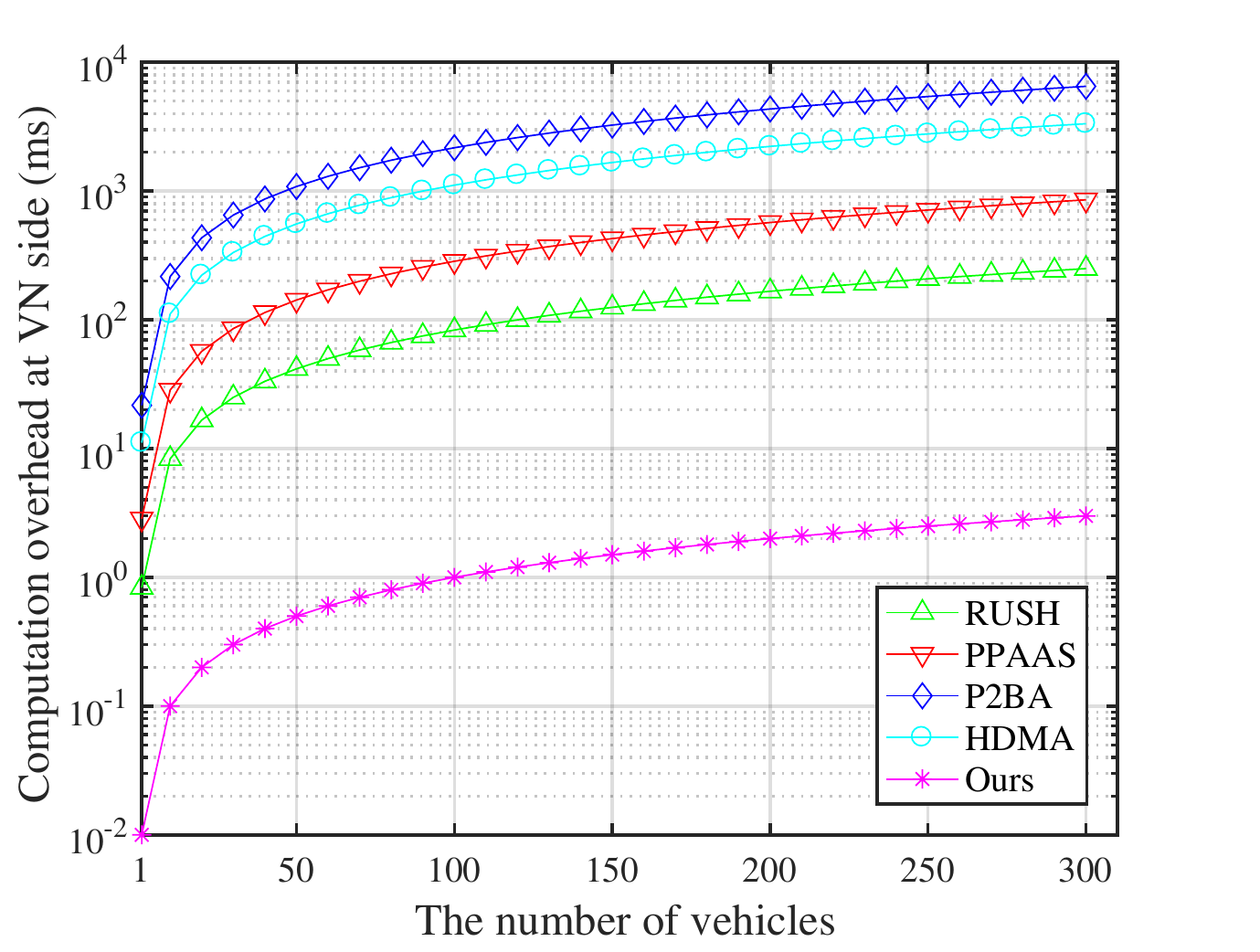"}
\caption{Computational overhead on VN side.}
\label{Computation on VN}
\end{figure}

\begin{figure}[!t]
\centering
\includegraphics[width=2.5in]{"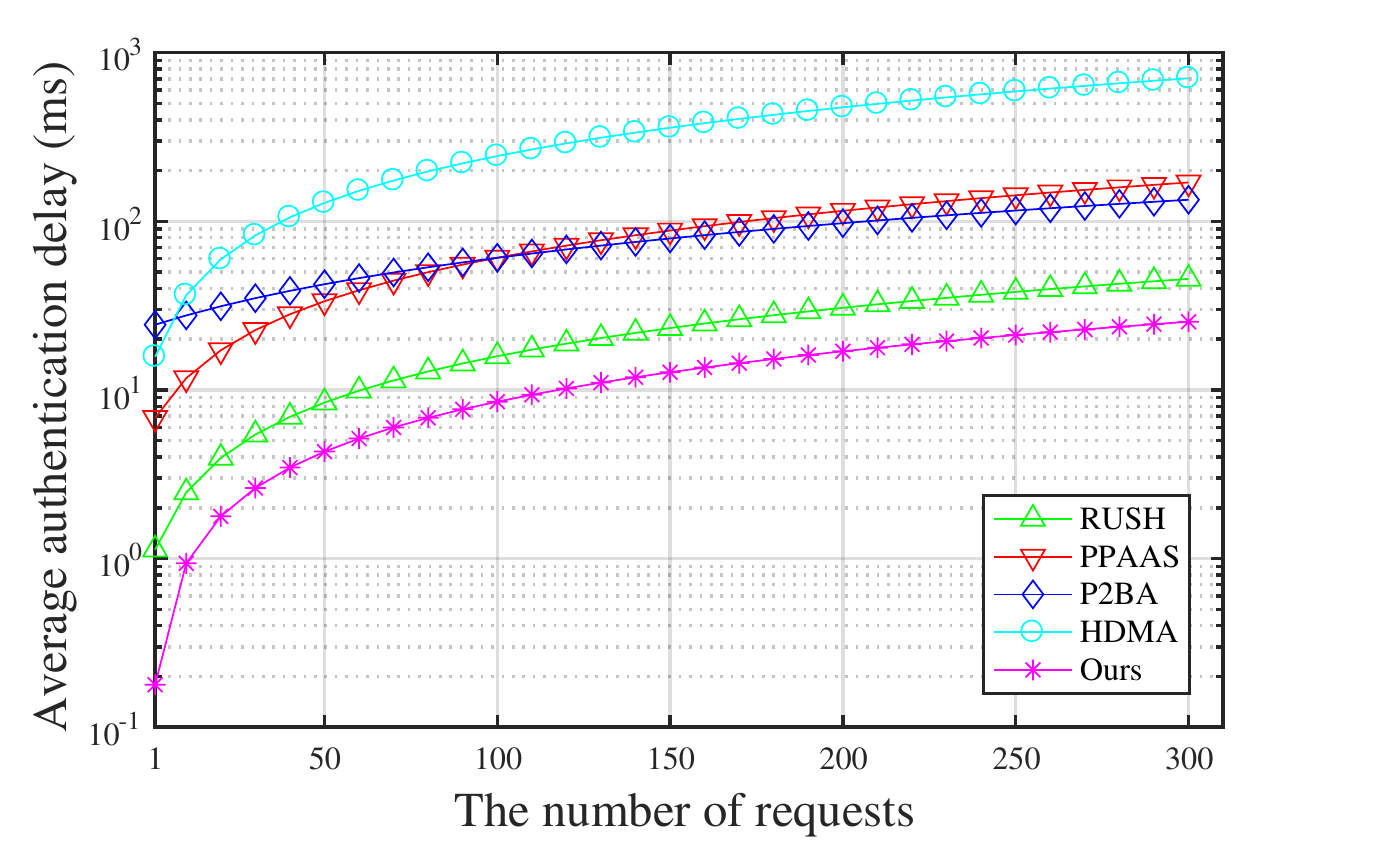"}
\caption{Average authentication delay versus the number of requests.}
\label{average authentication delay}
\end{figure}

\begin{figure}[!t]
\centering
\includegraphics[width=2.5in]{"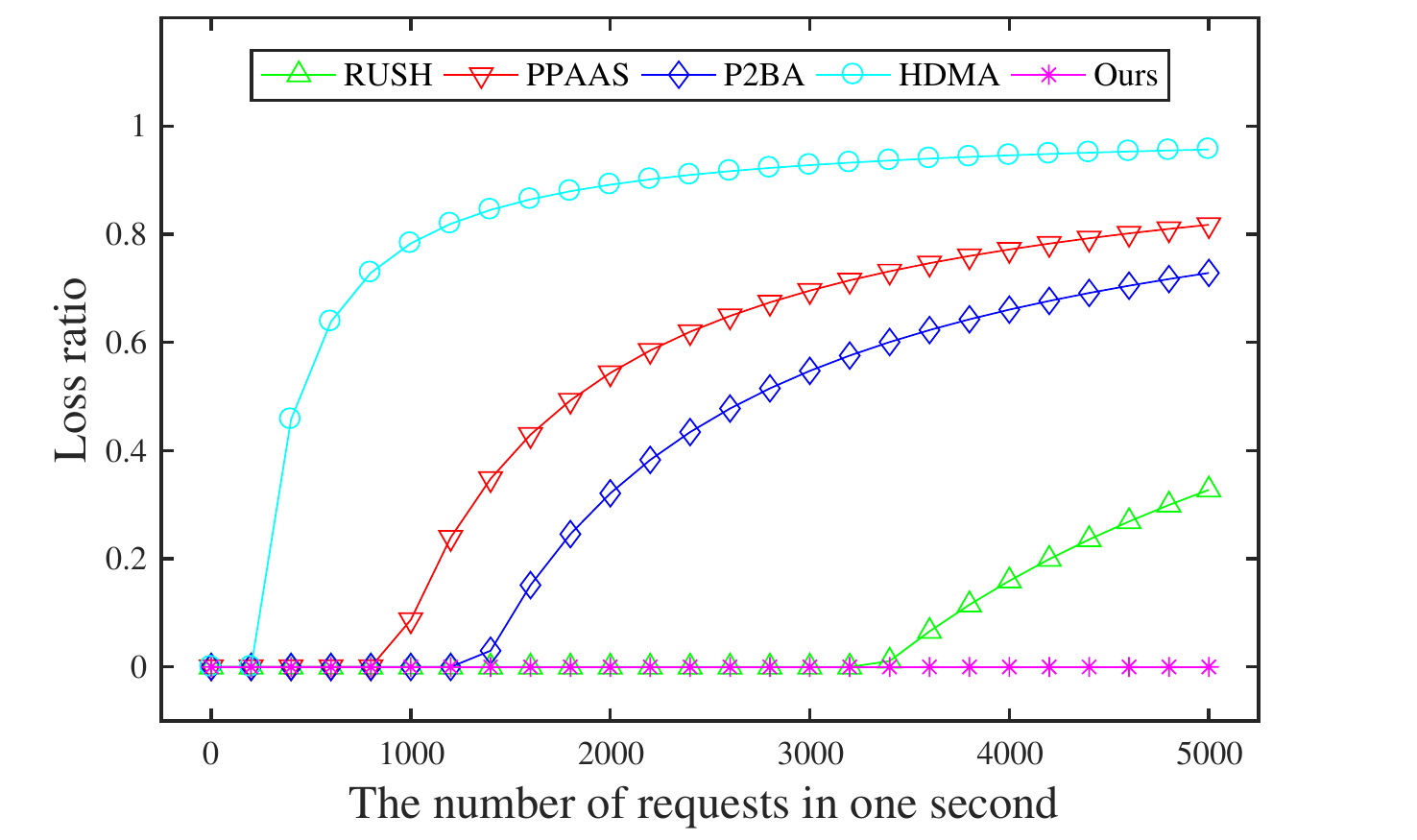"}
\caption{Loss ratio of requests on RSU.}
\label{Loss ratio}
\end{figure}

\subsection{Communication Cost}

In this section, we evaluate the communication overhead of BEPHAP. Let $q$ be a prime of length $l_q=224$ bits and $p$ be a prime of length $l_p=2048$ bits, which can keep up with the strength requirement since a 224-bit ECC key provides almost the same security level as a 2048-bit RSA key \cite{gallagher2013digital}. Note that $l_{pID}=16$ bytes, $l_{ID}=16$ bytes, and $l_{ts}=4$ bytes denote the length of the pseudo-identity, the identity and the timestamp, respectively. We adopt the hash function SHA3 whose output length is 160 bits. The signature size of P2BA \cite{DBLP:journals/tifs/FengSXW21} $l_{P2BA}$ is 768 bytes. As for PPAAS \cite{DBLP:journals/tifs/YangZZCZ22}, let $l_{G_1}$/$l_m$ denote the length of an element in $\mathbb{G}_1$/a message $m$ in PPAAS \cite{DBLP:journals/tifs/YangZZCZ22}. Since it is stated in \cite{DBLP:journals/tifs/YangZZCZ22}, $l_{G_1}$ is 77 bytes, $l_m$ is 100 bytes, we can get the length of a single ciphertext in \cite{DBLP:journals/tifs/YangZZCZ22} is 347 bytes. The message sizes of each scheme are given in Table \ref{communication compaison}. The total size of messages in BEPHAP is 212 bytes, less than those of the other schemes.

\begin{table*}
\renewcommand\tabcolsep{4.0pt} 
\centering
\caption{Comparison of Message Sizes (byte)}
\label{communication compaison}
\begin{threeparttable}
\resizebox{1.5\columnwidth}{!}{
\begin{tabular}{l|c|c|c|c}
\toprule
Scheme & Message1 & Message2 & Message3 & Total Message \\
\midrule
RUSH\cite{DBLP:journals/tdsc/0002DB021} & $3l_{q} + l_{pID} + l_{ts} = 104$ & $3l_q + l_{ID} + l_{ts} + l_h = 124$ & $l_h=20$ & 248 \\

PPAAS\cite{DBLP:journals/tifs/YangZZCZ22} & $3l_{G_1} + l_m + l_{ID} = 347$ & N/A & N/A & 347 \\

P2BA\cite{DBLP:journals/tifs/FengSXW21} & $l_{P2BA}=768$ & N/A & N/A & 768 \\

HDMA\cite{DBLP:journals/tits/WangCKSTL21} & $5l_p = 1280$ & $11l_p+2l_{ts} = 2824$ & N/A & 4104 \\
BEPHAP & $3l_q+l_{pID}+l_{ts}=104$ & $l_q+l_{pID}+l_{ts}+2l_h=88$ & $l_h=20$ & 212 \\
\bottomrule
\end{tabular}
}

\end{threeparttable}
  
\end{table*}

\subsection{Message Loss Ratio}
We use an Intel i5-6200U CPU @ 2.40 GHz as the VN and an Intel i9-10900 CPU @ 2.81 GHz as the $RSU_t$, then simulate the authentication process to calculate the message loss ratio.
We compare the message loss ratio \cite{DBLP:journals/tits/WangCKSTL21} of each scheme to further evaluate BEPHAP. 
The loss ratio is the ratio between the number of dropped requests and total requests within a fixed interval (1000 ms here). Fig.\ref{Loss ratio} shows the loss ratio on RSU in various schemes for handling the requests. It can be observed that when the number of requests is not greater than 5000, the message loss rate in BEPHAP remains zero, which is the best performance among the related schemes.

\section{Conclusion\label{Conclusion}}
In this paper, we present a Blockchain-based Efficient Privacy-preserving Handover Authentication Protocol with key agreement (BEPHAP) for IoV under a security model in which cloud servers and RSUs may launch attacks. Leveraging blockchain, symmetric cryptography, and chameleon hash, we can implement cross-domain privacy-preserving handover authentication. To the best of our knowledge, BEPHAP is the first blockchain-based authentication protocol scheme for IoV that simultaneously implement mutual authentication with key agreement, data confidentiality, identity anonymity, unlinkability, traceability, non-repudiation, non-frameability, key escrow freeness, cross-domain, formal security proof, and verification by formal security verification tools. BEPHAP is particularly suitable for IoV scenarios with constrained vehicle computing capabilities since vehicles in BEPHAP only need to perform lightweight cryptographic operations in the authentication phase, such as symmetric encryption and hash. The experiments demonstrate that an individual RSU can handle an authentication request within 0.2 ms, which is more efficient than the existing schemes. And it is worth noting that the computational cost of VN in BEPHAP is reduced by two or even three orders of magnitude compared to other schemes.

\bibliographystyle{./IEEEtran}

\begin{IEEEbiography}[{\includegraphics[width=1in,height=1.25in,clip,keepaspectratio]{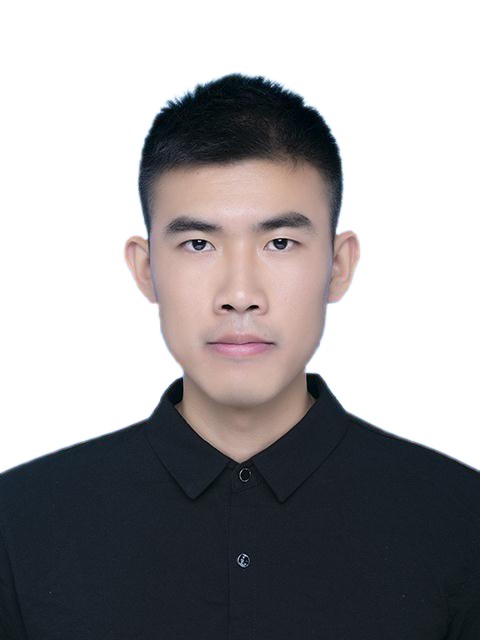}}]{Xianwang Xie} received his B.E. degree in mechatronics engineering from Zhejiang Sci-Tech University, Hangzhou, China, in 2020. He is currently working toward the M.E. degree in Electronic and Information Engineering with the Institute of Information Engineering, CAS, Beijing, China. His research interests include blockchain, network security, the internet of vehicles, etc.
\end{IEEEbiography}

\begin{IEEEbiography}[{\includegraphics[width=1in,height=1.25in,clip,keepaspectratio]{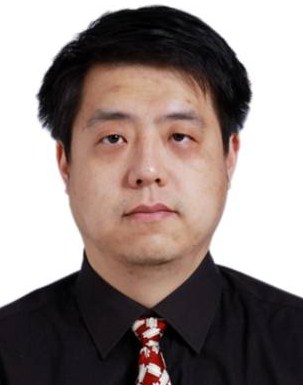}}]{Bin Wu} received his BS degree in automation and MS degree in computer science from the Ocean University of China in 2003 and 2006, respectively. He received his Ph.D. degree in information security from the Graduate University of Chinese Academy of Sciences in 2010. Now, he is an associate professor in State Key Laboratory of Information Security, Institute of Information Engineering, Chinese Academy of Sciences. His research interests include network security, covert communication, and blockchain.
\end{IEEEbiography}

\begin{IEEEbiography}[{\includegraphics[width=1in,height=1.25in,clip,keepaspectratio]{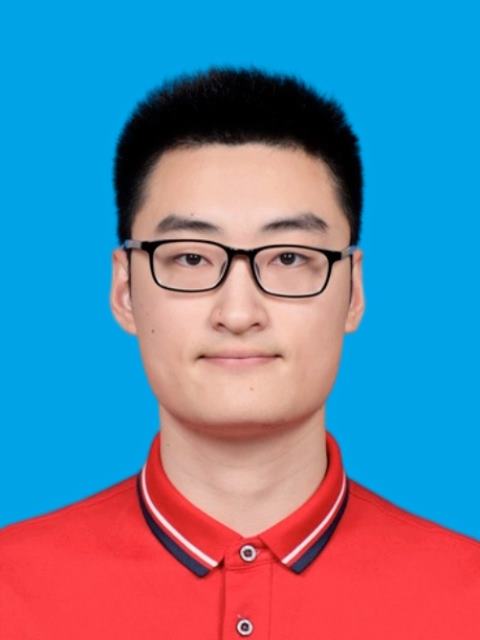}}]{Botao Hou} received the B.Sci. degree in Information Security from Shandong University, Shandong, China, in 2018. He is currently pursuing the PHD’s degree with the Department of Cyber Security, Institute of Information Engineering, CAS, Beijing, China. His research interests include blockchain applications and traffic analysis.
\end{IEEEbiography}

\end{document}